\documentclass[11pt,english]{paper}
\usepackage{lmodern}

\usepackage{courier}
\usepackage[T1]{fontenc}
\usepackage[latin9]{inputenc}
\usepackage[letterpaper]{geometry}
\usepackage{prettyref}
\usepackage{amsthm}
\usepackage{amsmath}
\usepackage{amssymb}
\usepackage{graphicx}
\usepackage{esint}

\makeatletter

\providecommand{\tabularnewline}{\\}

\numberwithin{equation}{section}
\numberwithin{figure}{section}
\newcommand{\lyxaddress}[1]{
\par {\raggedright #1
\vspace{1.4em}
\noindent\par}
}
\theoremstyle{plain}
\newtheorem{thm}{\protect\theoremname}
  \theoremstyle{plain}
  \newtheorem{lem}[thm]{\protect\lemmaname}

\@ifundefined{showcaptionsetup}{}{%
 \PassOptionsToPackage{caption=false}{subfig}}
\usepackage{subfig}
\makeatother

\usepackage{babel}
  \providecommand{\lemmaname}{Lemma}
\providecommand{\theoremname}{Theorem}

\begin{document}

\title{Navier-Stokes solver using Green's functions I: channel flow and
plane Couette flow}

\author{Divakar Viswanath and Ian Tobasco}

\maketitle

\lyxaddress{Department of Mathematics, University of Michigan (divakar@umich.edu)
and Courant Institute, New York University (ian.tobasco@gmail.com). }
\begin{abstract}
Numerical solvers of the incompressible Navier-Stokes equations have
reproduced turbulence phenomena such as the law of the wall, the dependence
of turbulence intensities on the Reynolds number, and experimentally
observed properties of turbulence energy production. In this article,
we begin a sequence of investigations whose eventual aim is to derive
and implement numerical solvers that can reach higher Reynolds numbers
than is currently possible. Every time step of a Navier-Stokes solver
in effect solves a linear boundary value problem. The use of Green's
functions leads to numerical solvers which are highly accurate in
resolving the boundary layer, which is a source of delicate but exceedingly
important physical effects at high Reynolds numbers. The use of Green's
functions brings with it a need for careful quadrature rules and a
reconsideration of time steppers. We derive and implement Green's
function based solvers for the channel flow and plane Couette flow
geometries. The solvers are validated by reproducing turbulence phenomena
in good agreement with earlier simulations and experiment.
\end{abstract}

\section{Introduction}

The incompressible Navier-Stokes equations are given by $\partial\mathbf{u}/\partial t+(\mathbf{u}.\mathbf{\nabla})\mathbf{u}=-\nabla p+\triangle\mathbf{u}/Re$,
where $\mathbf{u}$ is the velocity field and $p$ is pressure. The
incompressibility constraint is $\nabla.\mathbf{u}=0$. We assume
that a characteristic speed $U$ and a characteristic length $L$
have been chosen and that the Reynolds number $Re$ is given by $UL/\nu$,
where $\nu$ is the kinematic viscosity. It is assumed that the unit
for mass is chosen so that the fluid has  constant density equal to
$1$.

The topic of this paper is the use of Green's functions to solve the
incompressible Navier-Stokes equations. The Navier-Stokes equations
are nonlinear while Green's functions are based on linear superposition.
Thus the solutions of the incompressible Navier-Stokes equations cannot
be described using Green's functions. However, if we discretize the
Navier-Stokes equations in time but not in space, and the time discretization
treats the nonlinear advection term \texttt{$(\mathbf{u.\nabla)}\mathbf{u}$}
explicitly and the pressure term $-\nabla p$ and the viscous diffusion
term $\triangle\mathbf{u}/Re$ implicitly, each time step is a linear
boundary value problem. The simplest such discretization, which is
to treat the advection term using forward Euler and the pressure and
diffusion term using backward Euler, gives the equation
\[
\frac{\mathbf{u}^{n+1}-\mathbf{u}^{n}}{\Delta t}+\left((\mathbf{u}.\nabla)\mathbf{u}\right)^{n}=-\nabla p^{n+1}+\frac{1}{Re}\triangle\mathbf{u}^{n+1}
\]
where the superscripts indicate the time step. This is a linear boundary
value problem for $\mathbf{u}^{n+1}$ with the constraint $\nabla.\mathbf{u}^{n+1}=0$
and with boundary conditions on $\mathbf{u}$ depending upon the geometry
of the flow. Green's functions may be derived for this linear boundary
value problem as shown in the theory of hydrodynamic potentials \cite{Ladyzhenskaya1969}. 

Green's functions exploit the principle of linear superposition to
express solutions of linear boundary value problems in integral form.
In numerical methods based on Green's functions, the weight of the
method falls upon quadrature rules as opposed to rules for the discretization
of derivatives. The importance of quadrature rules is already clear
in the early work of Rokhlin \cite{Rokhlin1983,Rokhlin1985}, where
Richardson extrapolation and trapezoidal rules are used to effect
accurate quadrature of the integral equations of acoustic scattering
and potential theory.

From the beginning, Greengard, Rokhlin, and others \cite{GreengardRokhlin1991,Rokhlin1983,Rokhlin1985}
have emphasized the ability of Green's function based methods in handling
very thin boundary layers. Shear flows such as channel flow or pipe
flow or plane Couette flow are characterized by very thin boundary
layers at high Reynolds numbers. There is turbulence activity in the
boundary layer as well as in the outer flow and the viscous effects
propagate into the domain from the boundary layer. Green's function
based methods are likely to be advantageous in handling such boundary
layers. It is legitimate to ask why a numerical method must be believed
to capture the effect of the viscous term with the very small $1/Re$
coefficient. In Green's function based methods, that effect is captured
exactly by the analytic form of the Green's function. 

As far as we are aware, time integration using Green's functions has
not been tried on a nonlinear problem of the scale and difficulty
associated with fully developed turbulence. Thus some of the issues
that come up in relation to time integration in Section 3 cannot be
considered unexpected.

Many of the subtleties associated with the numerical integration of
the Navier-Stokes equations are related to the treatment of pressure.
The equations do not explicitly determine the evolution of pressure.
Instead, pressure is determined implicitly through the incompressibility
constraint on the velocity field. One of the key algorithms for solving
the Navier-Stokes equations in channel and plane Couette geometries
is due to Kleiser and Schumann \cite[1980]{KleiserSchumann1980}.
Kleiser and Schumann introduced a numerical technique for enforcing
the physically correct boundary conditions on pressure. Another method
was introduced by Kim, Moin, and Moser \cite[1987]{KimMoinMoser1987}
in a paper that is a landmark in the modern development of fluid mechanics.
Kim et al. \negthinspace{}\negthinspace{}\negthinspace{}\negthinspace{}
reproduced several features of fully developed turbulence from direct
numerical simulation of the Navier-Stokes equations. Their calculation
was initialized with a velocity field that was generated using large
eddy simulation. One of the highlights of the paper by Kim et al.
\negthinspace{}\negthinspace{}\negthinspace{}\negthinspace{} is
the correction of a calibration error in a published experiment using
numerical data.

\begin{figure}
\begin{centering}
\includegraphics[scale=0.5]{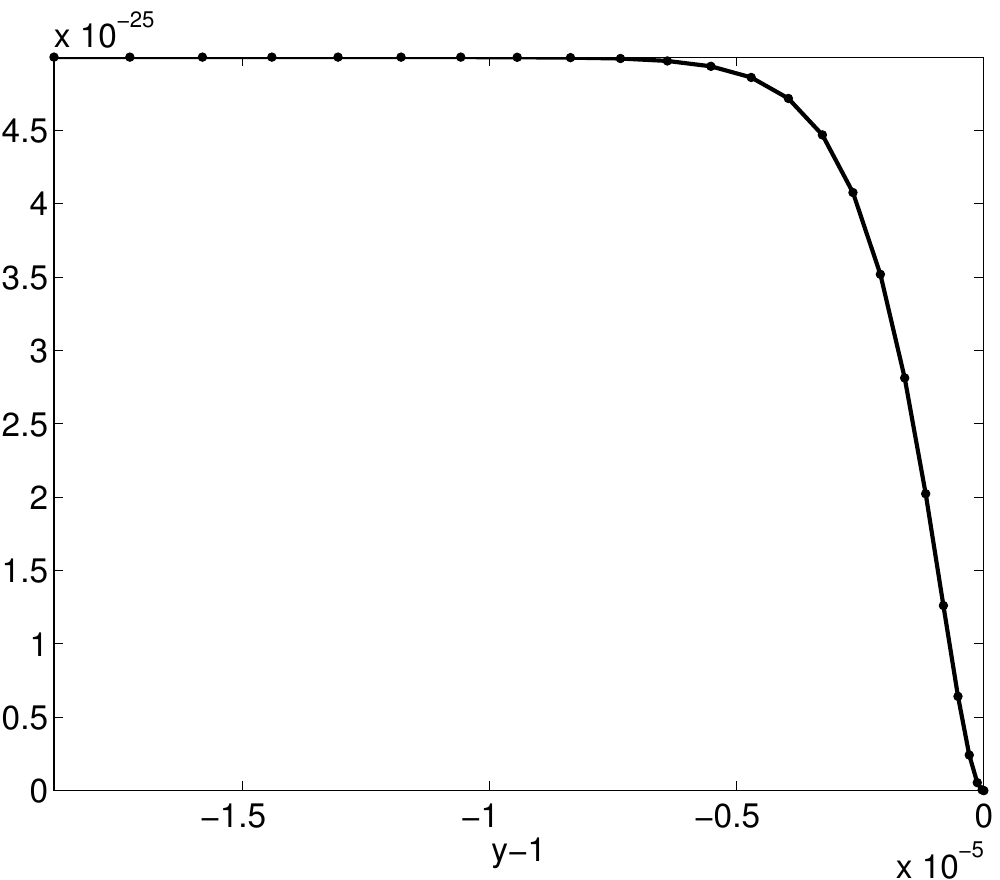}
\par\end{centering}

\caption{Very thin boundary layer at $y=1$ in the solution of the fourth order
boundary value problem $(D^{2}-\beta^{2})(D^{2}-\alpha^{2})u=1$ with
boundary conditions $u(\pm1)=u'(\pm1)=0$ and with parameters $\alpha=10^{6}$
and $\beta=\sqrt{2}\alpha$. The solid markers are from an exact formula
and the solid line is a numerical solution.\label{fig:Very-thin-boundary}}
\end{figure}
The channel geometry is rectangular with $x$, $y$, and $z$ being
the streamwise, wall-normal, and spanwise directions by convention.
The corresponding components of the velocity field $\mathbf{u}$ are
denoted $u$, $v$, and $w$. The walls are at $y=\pm1$ with fluid
in between. The no-slip boundary conditions require $\mathbf{u}=0$
at $y=\pm1$. The boundary conditions in the wall-parallel directions
are typically periodic in numerical work. The flow is driven either
by maintaining a constant mass flux or a constant pressure gradient
in the streamwise direction. In plane Couette flow, the geometry is
the same but the walls are moving. The no-slip boundary conditions
are $(u,v,w)=(0,\pm1,0)$ at $y=\pm1$. Plane Couette flow is driven
by the motion of the walls.

Both the Kleiser-Schumann and Kim-Moin-Moser methods come down to
solving linear boundary value problems in the $y$ or wall-normal
directions. The periodic directions are tackled using Fourier analysis
and dealiasing of the nonlinear advection term. Each Fourier component
then yields a linear boundary value problem in the $y$ direction.
The $y$ direction is discretized using Chebyshev points $y_{j}=\cos j\pi/M$
with $j=0,\ldots,M$. 

Here we parenthetically mention the interpretation of the parameters
$\alpha$ and $\beta$, which occur in the ensuing discussion. The
parameters are given by $\alpha^{2}=\ell^{2}/\Lambda_{x}^{2}+n^{2}/\Lambda_{z}^{2}$
and $\beta^{2}=\alpha^{2}+\gamma Re/\Delta t$, where $\ell$ and
$n$ are the Fourier modes and $2\pi\Lambda_{x}$ and $2\pi\Lambda_{z}$
are the dimensions of the domain in the streamwise and spanwise directions,
respectively. The parameter $\gamma$ depends on the time integration
scheme. More details are found in Section 3.

In Figure \ref{fig:Very-thin-boundary}, we have shown the solution
of the linear boundary value problem 
\begin{equation}
(D^{2}-\beta^{2})(D^{2}-\alpha^{2})u(y)=f(y)\quad\quad u(\pm1)=u'(\pm1)=0\label{eq:fourth-order-bvp}
\end{equation}
with $f(y)\equiv1$. Here $D=\frac{d}{dy}$. A fourth order boundary
value problem of this type occurs explicitly in the method of Kim-Moin-Moser
but it is treated as a composition of two second order boundary value
problems corresponding to the factors $D^{2}-\alpha^{2}$ and $D^{2}-\beta^{2}$.
In the method of Kleiser-Schumann, a fourth order boundary value problem
is not formed explicitly. Both methods solve the second order boundary
value problem 
\begin{equation}
(D^{2}-\beta^{2})u(y)=f(y)\quad\quad u(\pm1)=0\label{eq:second-order-bvp}
\end{equation}
by using the Chebyshev series $u(y)=\sum_{m=0}^{M}c_{m}T_{m}(y)$
and the set of equations obeyed by the coefficients $c_{n}$ given
on page 119 of Gottlieb and Orszag \cite[1977]{GottliebOrszag1977}. 

Although the method on p.\negthinspace{}\negthinspace{}\negthinspace{}
119 of Gottlieb and Orszag \cite{GottliebOrszag1977} has been extensively
used in turbulence computations for more than two decades, its numerical
properties have not been investigated as far as we know. For reliable
use in solving the Navier-Stokes equations at high Reynolds numbers,
the method should be able to accurately reproduce thin boundary layers,
such as the one shown in Figure \ref{fig:Very-thin-boundary}. There
is reason to be concerned. If the method is used to solve fourth order
problems of the type \prettyref{eq:fourth-order-bvp}, it forms linear
systems with condition numbers of the order $\alpha^{2}\beta^{2}$
for the fourth order boundary value problem \prettyref{eq:fourth-order-bvp}
and of order $\alpha^{2}$ for the second order boundary value problem
\prettyref{eq:second-order-bvp} \cite{Viswanath2012} . For the problem
shown in Figure \ref{fig:Very-thin-boundary}, the condition number
is more than $10^{24}$ and greatly exceeds the machine epsilon of
double precision arithmetic. 

Zebib \cite[1984]{Zebib1984} and Greengard \cite{Greengard1991}
suggested using a Chebyshev series for the highest derivative. For
the fourth order problem \prettyref{eq:fourth-order-bvp}, the Chebyshev
expansion would be $u''''=\sum_{m=0}^{M}c_{m}T_{m}(y)$. This device
avoids the ill-conditioning of Chebyshev differentiation due to clustering
at the end points that causes large errors in spectral differentiation.
However, the spectral integration method of Zebib and Greengard also
has a condition number greater than $\alpha^{2}\beta^{2}$ for the
fourth order boundary value problem \prettyref{eq:fourth-order-bvp}
\cite{CoutsiasHT1996,Viswanath2012}. 

The method of spectral integration has been extended and investigated
carefully in \cite{Viswanath2012}. The equations presented somewhat
tersely on p.\negthinspace{}\negthinspace{}\negthinspace{} 119
of Gottlieb and Orszag \cite{GottliebOrszag1977} are in fact a form
of spectral integration. The methods used by Kleiser-Schumann and
Kim-Moin-Moser have numerical properties that are practically identical
to that of Zebib and Greengard. The essential equivalence of the Gottlieb-Orszag
equations with spectral integration was first recognized by Charalambides
and Waleffe \cite{CharalambidesWaleffe2008}. Because of this equivalence
the advantages of explicit spectral integration, as implemented in
\cite{LundbladhHenningsonJohansson1992,LundbladhHenningsonReddy1994}
and \cite{Waleffe2003}, are not as overwhelming as illustrated in
Figure 3 of \cite{Waleffe2003}. When we refer to spectral integration,
it includes the methods of Gottlieb-Orszag, Kleiser-Schumann, Zebib,
Kim-Moin-Moser, and Greengard as well as the more general and powerful
formulations derived in \cite{Viswanath2012}. 

Regardless of which version of spectral integration is used, the fact
remains that the linear system for the fourth order boundary value
problem \prettyref{eq:fourth-order-bvp} has a condition number of
$\alpha^{2}\beta^{2}$ . Yet, remarkably, even systems with condition
numbers exceeding $10^{24}$ (see Figure \ref{fig:Very-thin-boundary})
can be solved with a loss of only five or six digits of accuracy.
The accuracy of spectral integration in spite of large condition numbers
can be partly explained using the singular value decomposition \cite{Viswanath2012}.
Another property of spectral integration (in all its forms) is that
some of the intermediate quantities have large errors which cancel
in the final answer \cite{Viswanath2012}. A robust implementation
must take these two properties into account. Spectral integration
can indeed handle thin boundary layers, such as the one shown in Figure
\ref{fig:Very-thin-boundary}, in spite of large condition numbers.
The robustness of spectral integration was essential to the outstanding
success of the methods of Kleiser-Schumann and Kim-Moin-Moser in more
than two decades of use (however, not all implementations are equal). 

In Figure \ref{fig:Very-thin-boundary}, the thickness of the boundary
layer is of the order $10^{-6}$. It takes more than $10,000$ Chebyshev
points in the interval $-1\leq y\leq1$ to resolve that boundary layer
in spite of quadratic clustering near the endpoints. That is a lot
more than the number needed if the grid points are chosen in a suitably
adaptive manner. Viswanath \cite{Viswanath2012} has derived a version
of spectral integration that applies to piecewise Chebyshev grid points.
Using that method, the number of grid points needed to solve a linear
boundary value with a boundary layer as thin as the one shown in Figure
\ref{fig:Very-thin-boundary} is reduced from $8192$ to $96$. It
appears that this new method can be used to obtain considerable improvement
in both the Kleiser-Schumann and Kim-Moin-Moser methods. 

The Green's function method, whose development we begin in this paper,
is an alternative which in its final form will enjoy the same advantages.
Spectral integration is an essentially one dimensional idea and cannot
be generalized to pipe flows with non-circular cross-sections and
to other non-rectangular geometries. The use of Green's functions
on the other hand will generalize. A great many analytic and numerical
complications arise when Green's functions are derived for cross-sections
of pipes as a part of a numerical method for solving the Navier-Stokes
equations. It is essential to develop the method for the channel geometry,
as we do here, before those difficulties are confronted. It is also
possible that the Green's function method will turn out to be faster
than the methods of Kleiser-Schumann and Kim-Moin-Moser, revised in
the manner suggested in the previous paragraph, but one cannot be
certain until the two alternatives are developed to their final form.
Spectral integration over piecewise Chebyshev grids appears to be
sensitive to the location of the nodes used to divided the interval
\cite{Viswanath2012}. The Green's function method is likely to be
much less sensitive. Lastly, we mention that the Green's function
method has a theoretical advantage. When implemented using suitable
quadrature rules, its numerical stability is immediately obvious. 

The Green's function method for solving the Navier-Stokes equations
in channel and plane Couette flow geometries is developed in Sections
2 and 3. The quadrature rule that is used is provisional. The way
to derive robust quadrature rules is indicated in Section 3 and the
complete method will be given in the sequel to this paper.

In Section 4, we show the theoretically intriguing result that Green's
functions may be used to eliminate numerical differentiation in the
wall-normal or $y$ direction entirely from the numerical scheme.
The derivatives that occur in the nonlinear advection term can be
transferred to the Green's function using integration by parts with
the result that numerical differentiation is replaced by analytic
differentiation. Such a scheme is not practical at high Reynolds numbers
for reasons given in that section.

In Section 5, we validate the Green's function based method. In view
of the extensions discussed in this introduction, the code has been
written in such a way that it can reach hundreds of millions of grid
points using only a dozen or two processor cores. Since the piecewise
Chebyshev extension with robust quadrature rules is yet to be fully
developed, the full capabilities of this code are not exercised. Yet
we report simulations with up to ten million grid points and investigate
certain aspects of fully developed turbulence to demonstrate the viability
of the Green's function approach.

\section{Green's functions and template boundary value problems}

Every time step in the solution of the Navier-Stokes equations in
the channel geometry reduces to the solution of a number of linear
boundary value problems of the type \eqref{eq:fourth-order-bvp} and
\eqref{eq:second-order-bvp}. In this section, we derive the Green's
functions for the solutions of those boundary value problems. The
Green's functions can be derived using very standard methods. However,
the resulting expressions are unsuitable for numerical evaluation.
When the parameters $\alpha$ and $\beta$ are as large as $10^{6}$,
as in Figure \ref{fig:Very-thin-boundary}, quantities of the type
${\rm e}^{\beta y}$ or ${\rm e}^{-\beta y}$ will overflow. Thus
we begin by deriving the Green's functions in a manner that leads
to expressions suitable for accurate numerical evaluation. In the
last part of this section, we consider the evaluation of derivatives
such as $du/dy$, where $u$ is the solution of either of the boundary
value problems \eqref{eq:fourth-order-bvp} and \eqref{eq:second-order-bvp},
and the evaluation of the solution $u$ when the source term $f$
is given in the form $f\equiv df_{1}/dy$.

\subsection{Green's functions of linear boundary value problems}

Let $Lu=u^{(n)}+a_{1}(y)u^{(n-1)}+\cdots+a_{n-1}(y)u^{(1)}+a_{n}(y)u$.
The coefficients $a_{i}(y)$, $1\leq i\leq n$, are assumed to be
\emph{real-valued} and sufficiently smooth. The adjoint operator is
given by $L^{+}v=(-1)^{n}v^{(n)}+(-1)^{n-1}(a_{1}v)^{(n-1)}+\cdots+a_{n}v$.
We assume throughout that the functions that arise have the requisite
order of smoothness and that $n\geq2$. The degree of differentiability
is specifically mentioned only if there is a nontrivial reason for
doing so.

The lemmas in this subsection are not new. They can be found in \cite{CoddingtonLevinson1955}
in some form or the other. However, our derivation leads to formulas
which are easier to manipulate and which are suitable for numerical
evaluation. Our derivation of the Green's function for the boundary
value problem $Lu=f$, $a\leq y\leq c$, with suitable boundary conditions
on $u$, is based on the Lagrange identity, which is the next lemma. 
\begin{lem}
For any two functions $u$ and $v$, the Lagrange identity $v\, Lu-u\, L^{+}v=[uv]'$
holds, with 
\[
[uv]=\sum_{k=0}^{n-1}\sum_{r=0}^{n-k-1}(-1)^{r}(va_{k})^{(r)}u^{(n-k-r-1)}
\]
and $a_{0}\equiv1$.\end{lem}
\begin{proof}
Begin with $\int v\, Lu\, dy$ and integrate each term by parts repeatedly.
\end{proof}
Define 
\[
\tilde{u}=\left(\begin{array}{c}
u\\
u^{(1)}\\
\vdots\\
u^{(n-1)}
\end{array}\right).
\]
The quantity $[uv]$ which appears in the Lagrange identity may be
written as $[uv]=\tilde{u}^{T}A\tilde{v}$, where $A$ is an $n\times n$
matrix. All the entries of $A$ are determined by the lemma. However,
all that we need to know about $A$ is that it has the following reverse
triangular structure
\[
A=\left(\begin{array}{cccc}
\cdot & \cdot & \cdot & (-1)^{n-1}\\
\cdot & \cdot & \cdot & 0\\
\cdot & -1 & 0 & 0\\
1 & 0 & 0 & 0
\end{array}\right)
\]
and that the reverse diagonal is as shown above.

Let $u_{1},\ldots,u_{n}$ be a basis of solutions of the homogeneous
problem $Lu=0$. Similarly, let $v_{1},\ldots,v_{n}$ be a basis of
solutions of the adjoint problem $L^{+}u=0$. Denote the $n\times n$
matrices 
\[
(\tilde{u}_{1},\ldots,\tilde{u}_{n})\quad\text{and}\quad(\tilde{v}_{1},\ldots,\tilde{v}_{n})
\]
by $U$ and $V$, respectively (the determinant of $U$ is the Wronskian).
The Lagrange identity implies the following lemma.
\begin{lem}
$\frac{d}{dy}\left(U^{T}AV\right)=0.$
\end{lem}
We will assume that the bases of solutions are chosen in such a way
that 
\begin{equation}
U(y)^{T}A(y)V(y)=I\label{eq:UTAVeqI}
\end{equation}
 where $I$ is the identity matrix. The homogeneous solutions $u_{i}$
and $v_{i}$ are used to construct the Green's function of $Lu=f$.
Before deriving the Green's function, we give the identity 
\begin{equation}
v_{j}=\frac{\det\left(\begin{array}{cccc}
u_{1} & \ldots & u_{1}^{(n-2)} & 0\\
u_{2} & \ldots & u_{2}^{(n-2)} & \cdot\\
 & \cdots &  & 1\\
 & \cdots &  & \cdot\\
u_{n} & \ldots & u_{n}^{(n-2)} & 0
\end{array}\right)}{\det\left(\begin{array}{cccc}
u_{1} & u_{1}' & \ldots & u_{1}^{(n-1)}\\
u_{2} & u_{2}' & \ldots & u_{2}^{(n-1)}\\
 &  & \cdots\\
u_{n} & u_{n}' & \ldots & u_{n}^{(n-1)}
\end{array}\right)}.\label{eq:vj-determinant}
\end{equation}
The entry equal to $1$ in the last column of the numerator is in
row number $j$. This identity is derived as follows. We choose the
$j$-th column of \eqref{eq:UTAVeqI} to get 
\[
U(y)^{T}A(y)\left(\begin{array}{c}
v_{j}\\
\vdots\\
v_{j}^{(n-1)}
\end{array}\right)=\left(\begin{array}{c}
0\\
\vdots\\
1\\
\vdots\\
0
\end{array}\right).
\]
Because of the reverse triangular structure of $A$, the last entry
of $A\tilde{v}_{j}$ is equal to $v_{j}$. Identity \eqref{eq:vj-determinant}
is implied by Cramer's rule. By working with rows of \eqref{eq:UTAVeqI}
instead of columns, we get the identity 
\begin{equation}
u_{j}=(-1)^{n-1}\frac{\det\left(\begin{array}{cccc}
v_{1} & \ldots & v_{1}^{(n-2)} & 0\\
v_{2} & \ldots & v_{2}^{(n-2)} & \cdot\\
 & \cdots &  & 1\\
 & \cdots &  & \cdot\\
v_{n} & \ldots & v_{n}^{(n-2)} & 0
\end{array}\right)}{\det\left(\begin{array}{cccc}
v_{1} & v_{1}' & \ldots & v_{1}^{(n-1)}\\
v_{2} & v_{2}' & \ldots & v_{2}^{(n-1)}\\
 &  & \cdots\\
v_{n} & v_{n}' & \ldots & v_{n}^{(n-1)}
\end{array}\right)}.\label{eq:uj-determinant}
\end{equation}
The identities \eqref{eq:vj-determinant} and \eqref{eq:uj-determinant}
are used to construct Green's functions in Section 2.2. 

So far, we have not specified the boundary conditions $u$ must satisfy
in addition to $Lu=f$. We take the domain to be $a\leq y\leq c$
and require that $\tilde{u}(a)$ must lie in an $n-\ell$ dimensional
subspace $V_{\ell}$ (this corresponds to $\ell$ linear conditions
on $\tilde{u}(a)$). Similarly, the right boundary conditions require
that $\tilde{u}(c)$ should lie in a $n-r$ dimensional subspace $V_{r}$.
We require $\ell+r=n$. The Green's function is built up using homogeneous
solutions of $Lu=0$ which satisfy the left or the right boundary
conditions. We assume that the basis solutions are chosen and then
ordered in such a way that 
\[
\tilde{u}_{1}(c),\ldots,\tilde{u}_{\ell}(c)\quad\text{and}\quad\tilde{u}_{\ell+1}(a),\ldots,\tilde{u}_{n}(a)
\]
span the subspaces $V_{r}$ and $V_{l}$, respectively. The following
lemma gives the boundary conditions satisfied by $v_{1},\ldots,v_{n}$,
a basis of solutions of $L^{+}v=0$ which is related to $u_{1},\ldots,u_{n}$
by \eqref{eq:UTAVeqI}. The lemma is useful for checking correctness
of the implementation. It may also be used for the construction of
$v_{i}$ given $u_{i}$. Its proof is obvious from $U^{T}AV=I$.
\begin{lem}
$\tilde{v}_{1}(a),\ldots,\tilde{v}_{l}(a)$ span the orthogonal complement
of the $n-\ell$ dimensional space $A(a)^{T}V_{l}$ and $\tilde{v}_{l+1}(c),\ldots,\tilde{v}_{n}(c)$
span the orthogonal complement of the $n-r$ dimensional space $A(c)^{T}V_{r}$.\label{lem:lemma3}
\end{lem}
Let $u$ be the solution of $Lu=f$ subject to the boundary conditions
$\tilde{u}(a)\in V_{\ell}$ and $\tilde{u}(b)\in V_{r}$. If we apply
the Lagrange identity using $u$ and $v_{i}$, where $v_{i}$ is a
solution of the homogeneous problem $L^{+}v=0$, we get $fv_{i}=\frac{d}{dy}\left(\tilde{u}^{T}Av_{i}\right)$
for $i=1,\ldots,n$. The equations with $i=1,\ldots,\ell$ are integrated
from $a$ to $y$ and the rest are integrated from $y$ to $c$. The
boundary conditions as given by the previous lemma imply that 

\begin{eqnarray*}
\int_{a}^{y}fv_{1} & = & \tilde{u}(y)^{T}A(y)\tilde{v}_{1}(y)\\
 & \cdots\\
\int_{a}^{y}fv_{l} & = & \tilde{u}(y)^{T}A(y)\tilde{v}_{l}(y)\\
-\int_{y}^{c}fv_{l+1} & = & \tilde{u}(y)^{T}A(y)\tilde{v}_{l+1}(y)\\
 & \cdots\\
-\int_{y}^{c}fv_{n} & = & \tilde{u}(y)^{T}A(y)\tilde{v}_{n}(y).
\end{eqnarray*}
The last entry of $A(y)^{T}\tilde{u}$ is equal to $(-1)^{n}u$. Using
Cramer's rule, we get
\begin{equation}
(-1)^{n-1}u=\frac{\det\left(\begin{array}{ccccc}
v_{1} & v_{1}' & \ldots & v_{1}^{(n-2)} & \int_{a}^{y}fv_{1}\\
 &  & \cdots\\
v_{n} & v_{n}' & \ldots & v_{n}^{(n-2)} & -\int_{y}^{c}fv_{n}
\end{array}\right)}{\det\left(\begin{array}{cccc}
v_{1} & v_{1}' & \ldots\  & v_{1}^{(n-1)}\\
v_{2} & v_{2}' & \ldots\  & v_{2}^{(n-1)}\\
 &  & \cdots\\
v_{n} & v_{n}' & \ldots\  & v_{n}^{(n-1)}
\end{array}\right)}.\label{eq:greenfndetform}
\end{equation}
The following lemma gives the Green's function in a more useful form. 
\begin{lem}
The solution of $Lu=f$ subject to the boundary conditions $u(a)\in V_{\ell}$
and $u(c)\in V_{r}$ is given by 
\begin{multline}
u(y)=u_{1}(y)\int_{a}^{y}v_{1}(\eta)f(\eta)\, d\eta+\cdots+u_{l}(y)\int_{a}^{y}v_{l}(\eta)f(\eta)\, d\eta\\
-u_{l+1}(y)\int_{y}^{c}v_{l+1}(\eta)f(\eta)\, d\eta-\cdots-u_{n}(y)\int_{y}^{c}v_{n}(\eta)f(\eta)\, d\eta.\label{eq:Green's-function}
\end{multline}
\end{lem}
\begin{proof}
Use \eqref{eq:uj-determinant} and \eqref{eq:greenfndetform}.
\end{proof}
The following lemma justifies the delta-function interpretation of
Green's functions favored by physicists. It is used in Section 2.3
and Section 4.
\begin{lem}
Let $u_{1},\ldots,u_{n}$ and $v_{1},\ldots,v_{n}$ be bases of solutions
of the homogeneous problems $Lu=0$ and $L^{+}v=0$, respectively,
that are related by $U(y)^{T}A(y)V(y)=I$. Then we have
\begin{eqnarray*}
\sum_{i=1}^{n}u_{i}^{(j)}v_{i} & = & 0\quad\text{for}\quad j=0,1,\ldots,n-2\\
 & = & 1\quad\text{for}\quad j=n-1
\end{eqnarray*}
and
\begin{eqnarray*}
\sum_{i=1}^{n}u_{i}v_{i}^{(j)} & = & 0\quad\text{for}\quad j=0,1,\ldots,n-2\\
 & = & (-1)^{n-1}\quad\text{for}\quad j=n-1.
\end{eqnarray*}
\label{lem:lemma5}\end{lem}
\begin{proof}
Use \eqref{eq:vj-determinant} and \eqref{eq:uj-determinant}.
\end{proof}

\subsection{Template boundary value problems}

The first template boundary value problem is 
\[
\left(D^{2}-\beta^{2}\right)u=f
\]
with boundary conditions $u(\pm1)=0$. The Green's function of this
boundary value can be deduced in any number of ways. We have 
\begin{equation}
G(y,\eta)=\frac{-{\rm e}^{\beta\,\left(-2+y+\eta\right)}+{\rm e}^{-\beta\,\left(4-y+\eta\right)}+{\rm e}^{\beta\,\left(-y+\eta\right)}-{\rm e}^{-\beta\,\left(2+y+\eta\right)}}{2\,\beta\,\left({\rm e}^{-4\,\beta}-1\right)}\label{eq:green2}
\end{equation}
for $-1\leq\eta\leq y\leq1$. The Green's function is symmetric and
the solution is given by $u(y)=\int_{-1}^{y}G(y,\eta)\, f(\eta)\, d\eta+\int_{y}^{\eta}G(\eta,y)\, f(\eta)\, d\eta.$
This form of the Green's function suits numerical work because none
of the terms will overflow for even $\beta$ very large. The terms
in the numerator are factored as follows:
\begin{align*}
{\rm e}^{\beta(-2+y+\eta)} & ={\rm e}^{\beta(-1+y)}{\rm e}^{\beta(-1+\eta)}\\
{\rm e}^{-\beta(4-y+\eta)} & ={\rm e}^{-2\beta}{\rm e}^{-\beta(1-y)}{\rm e}^{-\beta(1+\eta)}\\
{\rm e}^{-\beta(2+y+\eta)} & ={\rm e}^{-\beta(1+y)}{\rm e}^{-\beta(1+\eta)}.
\end{align*}
None of the factors will overflow even for large $\beta$. The term
${\rm e}^{\beta(-y+\eta)}$ is not factored and will not overflow
because $\eta\leq y$. Using these factorizations and noting that
$y$ and $\eta$ must be exchanged to get the Green's function for
$y\leq\eta$, we infer that the evaluation of the solution of the
first template problem using $u=\int_{-1}^{y}G(y,\eta)\, f(\eta)\, d\eta+\int_{y}^{\eta}G(\eta,y)\, f(\eta)\, d\eta$
reduces to the quadratures
\begin{equation}
\int_{-1}^{y}{\rm e}^{-\mu(\eta+1)}f(\eta)\, d\eta,\:\int_{y}^{1}{\rm e}^{-\mu(\eta+1)}f(\eta)\, d\eta,\;\int_{-1}^{y}{\rm e}^{\mu(-1+\eta)}f(\eta)\, d\eta,\;\int_{y}^{1}{\rm e}^{\mu(-1+\eta)}f(\eta)\, d\eta\label{eq:quadrature-1}
\end{equation}
and 
\begin{equation}
\int_{-1}^{+1}{\rm e}^{-\mu|y-\eta|}f(\eta)\, d\eta\label{eq:quadrature-2}
\end{equation}
with $\mu=\beta$. Each one of these quadratures yields a function
of $y$ and must be multiplied by a prefactor which is also a function
$y.$ For example, the first term of \eqref{eq:green2} contributes
the prefactor $-{\rm e}^{\beta(-1+y)}/2\beta({\rm e}^{-4\beta}-1)$
to $\int_{-1}^{y}{\rm e}^{\beta(-1+\eta)}f(\eta)\, d\eta$. Since
the term is unchanged when $y$ and $\eta$ are interchanged to obtain
the Green's function in $y\leq\eta$ region, it contributes the same
prefactor to $\int_{y}^{1}{\rm e}^{\beta(-1+\eta)}f(\eta)\, d\eta$.
The prefactor of the function defined by \eqref{eq:quadrature-2}
with $\gamma=\beta$ is $1/2\beta({\rm e^{-4\beta}-1)}$. Unlike the
result of the quadratures \eqref{eq:quadrature-1} and \eqref{eq:quadrature-2},
the prefactors do not depend upon $f$ and can be computed and stored
in advance. Thus the cost of solving the first template problem is
very nearly equal to the cost of the quadratures \eqref{eq:quadrature-1}
and \eqref{eq:quadrature-2} with $\gamma=\beta$.

The second template problem is the fourth order boundary value problem
\[
\left(D^{2}-\beta^{2}\right)\left(D^{2}-\alpha^{2}\right)u=f
\]
 with boundary conditions $u(\pm1)=u'(\pm1)=0$. For this problem,
it takes more work to write the Green's function in such a way that
there are no numerical overflows even if $\alpha$ and $\beta$ are
very large. However, the final result is similar to what we have seen
for the first template problem. The evaluation of $u$ can be reduced
to the quadratures \eqref{eq:quadrature-1} and \eqref{eq:quadrature-2}
with $\gamma=\alpha$ and $\gamma=\beta$. The results of quadratures
are multiplied by prefactors and summed to obtain $u$. 

We now turn to the derivation of the $4\times4$ matrix shown in Figure
\ref{fig:EntriesOf4x4Matrix}. That matrix is useful for computing
the prefactors.

Like the first template problem, the second template problem is self-adjoint.
We take the basis of homogeneous solutions to be 
\begin{align*}
u_{1} & ={\rm e}^{\beta\,\left(y-1\right)}+{\rm e}^{-\beta\,\left(y-1\right)}-{\rm e}^{\alpha\,\left(y-1\right)}-{\rm e}^{-\alpha\,\left(y-1\right)}\\
u_{2} & =\alpha\,{{\rm e}^{\beta\,\left(y-1\right)}}-\alpha\,{{\rm e}^{-\beta\,\left(y-1\right)}}-\beta\,{{\rm e}^{\alpha\,\left(y-1\right)}}+\beta\,{{\rm e}^{-\alpha\,\left(y-1\right)}}\\
u_{3} & ={{\rm e}^{\beta\,\left(y+1\right)}}+{{\rm e}^{-\beta\,\left(y+1\right)}}-{{\rm e}^{\alpha\,\left(y+1\right)}}-{{\rm e}^{-\alpha\,\left(y+1\right)}}\\
u_{4} & =\alpha\,{{\rm e}^{\beta\,\left(y+1\right)}}-\alpha\,{{\rm e}^{-\beta\,\left(y+1\right)}}-\beta\,{{\rm e}^{\alpha\,\left(y+1\right)}}+\beta\,{{\rm e}^{-\alpha\,\left(y+1\right)}}.
\end{align*}
It may be verified that $u_{1}$ and $u_{2}$ satisfy the right boundary
conditions while $u_{3}$ and $u_{4}$ satisfy the left boundary conditions
as assumed in Section 2.1. The functions $v_{1},v_{2},v_{3},v_{4}$
may be calculated using \eqref{eq:vj-determinant} or Lemma \ref{lem:lemma3}.
It is convenient to define 
\[
W=-4\,\alpha\,\beta\,{\delta}^{2}{\sigma}^{2}{{\rm e}^{-4\,\alpha}}-4\,\alpha\,\beta\,{\delta}^{2}{\sigma}^{2}{{\rm e}^{-4\,\beta}}+4\,{\frac{\alpha\,\beta\,{\delta}^{4}{{\rm e}^{-4\,\beta-4\,\alpha}}}{{\sigma}^{2}}}+32\,{\alpha}^{2}{\beta}^{2}{\delta}^{2}{{\rm e}^{-2\,\beta-2\,\alpha}}+4\,{\frac{\alpha\,\beta\,{\delta}^{4}}{{\sigma}^{2}}}
\]
where $\delta=\alpha^{2}-\beta^{2}$ and $\sigma=\alpha+\beta$ (the
Wronskian is equal to ${\rm e}^{2\alpha+2\beta}W$). The function
$v_{1}$ is equal to{\small 
\begin{multline*}
2\,\alpha\,\beta\,\left(\beta+\alpha\right)\left(-\beta+\alpha\right)^{2}{{\rm e}^{\alpha\,\left(y-1\right)}}-2\,\alpha\,\beta\,\left(\beta+\alpha\right)\left(-\beta+\alpha\right)^{2}{{\rm e}^{\beta\,\left(y-1\right)}}+2\,\alpha\,\beta\,\left(-\beta+\alpha\right)\left(\beta+\alpha\right)^{2}{{\rm e}^{-\alpha\,\left(y+3\right)}}\\
+2\,\alpha\,\beta\,\left(-\beta+\alpha\right)\left(\beta+\alpha\right)^{2}{{\rm e}^{-\beta\,\left(y+3\right)}}+2\,\alpha\,\beta\,\left(\beta+\alpha\right)\left(-\beta+\alpha\right)^{2}{{\rm e}^{-4\,\alpha-\beta\, y-3\,\beta}}-2\,\alpha\,\beta\,\left(-\beta+\alpha\right)\left(\beta+\alpha\right)^{2}{{\rm e}^{-4\,\alpha+\beta\, y-\beta}}\\
-4\,{\alpha}^{2}\beta\,\left(-\beta+\alpha\right)\left(\beta+\alpha\right){{\rm e}^{-2\,\alpha-\beta\, y-\beta}}-2\,\alpha\,\beta\,\left(\beta+\alpha\right)\left(-\beta+\alpha\right)^{2}{{\rm e}^{-4\,\beta-\alpha\, y-3\,\alpha}}-2\,\alpha\,\beta\,\left(-\beta+\alpha\right)\left(\beta+\alpha\right)^{2}{{\rm e}^{-4\,\beta+\alpha\, y-\alpha}}\\
-4\,\alpha\,{\beta}^{2}\left(-\beta+\alpha\right)\left(\beta+\alpha\right){{\rm e}^{-2\,\beta-\alpha\, y-\alpha}}+4\,\alpha\,{\beta}^{2}\left(-\beta+\alpha\right)\left(\beta+\alpha\right){{\rm e}^{\alpha\, y-3\,\alpha-2\,\beta}}+4\,{\alpha}^{2}\beta\,\left(-\beta+\alpha\right)\left(\beta+\alpha\right){{\rm e}^{\beta\, y-3\,\beta-2\,\alpha}}
\end{multline*}
}divided by $W$. The expression for $v_{2}$ is similarly long.

{\scriptsize }
\begin{figure}
\begin{centering}
{\scriptsize 
\[
\left(\begin{array}{cccc}
-2\,\alpha\,{\delta}^{2}+2\,\alpha\,{\delta}^{2}{{\rm e}^{-4\,\alpha}} & -8\,\beta\,{\alpha}^{2}\delta\,{{\rm e}^{-2\,\beta-2\,\alpha}}+2\,\alpha\,\delta\,{\sigma}^{2}{{\rm e}^{-4\,\alpha}}-2\,{\frac{\alpha\,{\delta}^{3}}{{\sigma}^{2}}}\\
8\,\beta\,{\alpha}^{2}\delta\,{{\rm e}^{-2\,\alpha}}-2\,\alpha\,\delta\,{\sigma}^{2}{{\rm e}^{-2\,\beta}}+2\,{\frac{\alpha\,{\delta}^{3}{{\rm e}^{-4\,\alpha-2\,\beta}}}{{\sigma}^{2}}} & -2\,\alpha\,{\delta}^{2}+2\,\alpha\,{\delta}^{2}{{\rm e}^{-4\,\alpha}}\\
-4\,{\frac{\beta\,\alpha\,{\delta}^{2}{{\rm e}^{-2\,\beta-2\,\alpha}}}{\sigma}}+4\,{\frac{\beta\,\alpha\,{\delta}^{2}}{\sigma}} & -4\,\beta\,\alpha\,\delta\,\sigma\,{{\rm e}^{-2\,\alpha}}+4\,\beta\,\alpha\,\delta\,\sigma\,{{\rm e}^{-2\,\beta}}\\
-4\,\beta\,\alpha\,\delta\,\sigma\,{{\rm e}^{-2\,\alpha}}+4\,\beta\,\alpha\,\delta\,\sigma\,{{\rm e}^{-2\,\beta}} & -4\,{\frac{\beta\,\alpha\,{\delta}^{2}{{\rm e}^{-2\,\beta-2\,\alpha}}}{\sigma}}+4\,{\frac{\beta\,\alpha\,{\delta}^{2}}{\sigma}}
\end{array}\right)
\]
}
\par\end{centering}{\scriptsize \par}

\begin{centering}
{\scriptsize 
\[
\left(\begin{array}{cccc}
 &  & -4\,{\frac{\beta\,\alpha\,{\delta}^{2}{{\rm e}^{-2\,\beta-2\,\alpha}}}{\sigma}}+4\,{\frac{\beta\,\alpha\,{\delta}^{2}}{\sigma}} & -4\,\beta\,\alpha\,\delta\,\sigma\,{{\rm e}^{-2\,\alpha}}+4\,\beta\,\alpha\,\delta\,\sigma\,{{\rm e}^{-2\,\beta}}\\
 &  & -4\,\beta\,\alpha\,\delta\,\sigma\,{{\rm e}^{-2\,\alpha}}+4\,\beta\,\alpha\,\delta\,\sigma\,{{\rm e}^{-2\,\beta}} & -4\,{\frac{\beta\,\alpha\,{\delta}^{2}{{\rm e}^{-2\,\beta-2\,\alpha}}}{\sigma}}+4\,{\frac{\beta\,\alpha\,{\delta}^{2}}{\sigma}}\\
 &  & -2\,\beta\,{\delta}^{2}+2\,\beta\,{\delta}^{2}{{\rm e}^{-4\,\beta}} & 8\,{\beta}^{2}\alpha\,\delta\,{{\rm e}^{-2\,\beta-2\,\alpha}}+2\,{\frac{\beta\,{\delta}^{3}}{{\sigma}^{2}}}-2\,\beta\,\delta\,{\sigma}^{2}{{\rm e}^{-4\,\beta}}\\
 &  & -8\,{\beta}^{2}\alpha\,\delta\,{{\rm e}^{-2\,\beta}}-2\,{\frac{\beta\,{\delta}^{3}{{\rm e}^{-4\,\beta-2\,\alpha}}}{{\sigma}^{2}}}+2\,\beta\,\delta\,{\sigma}^{2}{{\rm e}^{-2\,\alpha}} & -2\,\beta\,{\delta}^{2}+2\,\beta\,{\delta}^{2}{{\rm e}^{-4\,\beta}}
\end{array}\right)
\]
}
\par\end{centering}{\scriptsize \par}

{\scriptsize \caption{Entries of a $4\times4$ matrix with the first two columns placed
above the last two. This matrix determines the Green's function of
the second template problem. Here $\delta=\alpha^{2}-\beta^{2}$ and
$\sigma=\alpha+\beta$.\label{fig:EntriesOf4x4Matrix}}
}
\end{figure}
{\scriptsize \par}

For $\eta\leq y$, the Green's function is given by $G(y,\eta)=u_{1}(y)v_{1}(\eta)+u_{2}(y)v_{2}(\eta).$
This Green's function is determined by the $4\times4$ matrix shown
in Figure \ref{fig:EntriesOf4x4Matrix}. We think of the rows and
columns of the matrix as being indexed by $-\beta,\beta,-\alpha,\alpha$
in that order. The $(-\beta,-\beta)$ entry, which appears in the
top left corner, is divided by $W$ to get the coefficient of ${\rm e}^{-\beta(y+1)}{\rm e}^{-\beta(\eta+1)}$
in the expression for $G(y,\eta)$ for $\eta\leq y$. The other entries
are interpreted similarly but there are two special entries. These
are the $(-\beta,\beta)$ entry which must be interpreted as $W$
times the coefficient of ${\rm e}^{\beta(\eta-y)}$ and the $(-\alpha,\alpha)$
entry which must be interpreted as $W$ times the coefficient of ${\rm e}^{\alpha(\eta-y)}$.
The Green's function for $y\leq\eta$ is obtained from symmetry.

It follows that solving the second template problem $\left(D^{2}-\beta^{2}\right)\left(D^{2}-\alpha^{2}\right)u=f$
with boundary conditions $u(\pm1)=u'(\pm1)=0$ reduces to quadratures
\eqref{eq:quadrature-1} and \eqref{eq:quadrature-2} with $\gamma=\alpha$
and $\gamma=\beta$. Each quadrature yields a function of $y$ which
is multiplied by a prefactor. The prefactor is determined using the
$4\times4$ matrix of Figure \ref{fig:EntriesOf4x4Matrix} and the
formula for $W$. 

The formula for $W$ and the entries of the $4\times4$ matrix use
$\delta$ and $\sigma$ to avoid cancellation errors. Because of the
way the parameters $\alpha$ and $\beta$ arise in the numerical integration
of channel flow, $\delta=\alpha^{2}-\beta^{2}$ can be evaluated accurately.

Each of the quadratures \eqref{eq:quadrature-1} and \eqref{eq:quadrature-2}
is well-conditioned. However, if $\alpha\approx\beta$ there will
be large cancellation errors when the results of the quadratures are
multiplied by prefactors and summed. This phenomenon may be understood
as follows. When $\alpha\neq\beta$, the solutions ${\rm e}^{\pm\alpha y},{\rm e}^{\pm\beta y}$
form a basis of homogeneous solutions. When $\alpha=\beta$, the basis
is ${\rm e}^{\pm\alpha y},y{\rm e}^{\pm\alpha y}$. When $\alpha\approx\beta$,
the Green's function tries to produce terms which resemble $y{\rm e}^{\alpha y}$
using terms such as $\left({\rm e}^{\alpha y}-{\rm e}^{\beta y}\right)/\left(\alpha-\beta\right)$
resulting in large cancellation errors. Fortunately, this situation
does not arise in channel flow or plane Couette flow.

\subsection{Derivatives using Green's functions}

For the template boundary value problems $Lu=f$, we have derived
Green's functions such that $u(y)=\int_{-1}^{1}G(y,\eta)\, f(\eta)\, d\eta$.
Here we will consider the use of Green's functions to evaluate derivatives
such as $du/dy.$

The ability to differentiate solutions of boundary value problems
using Green's functions has been utilized in an important paper by
Greengard and Rokhlin \cite{GreengardRokhlin1991}. They consider
the boundary value problem $u''+p(y)u'+q(y)u=f(y)$ and solve it by
representing the solution $u$ in the form $u=\int_{-1}^{1}G(y,\eta)\,\sigma(\eta)\, d\eta$,
where $G$ is the Green's function of a linear boundary value problem
with constant coefficients which satisfies the same boundary conditions.
In fact, the background boundary value problem is simply taken to
be $u''=f$. With the representation of $u$ using $\sigma$, the
boundary value problem becomes an integral equation for $\sigma$.
Starr and Rokhlin \cite{StarrRokhlin1994} have generalized the method
to first order systems. The papers by Greengard, Rokhlin, and Starr
show how to apply numerical methods based on Green's functions to
problems with non-constant coefficients. Once the boundary value problem
is cast in integral form using the background Green's function, the
method handles diagonal blocks using Nyström integration and pieces
together the global solution efficiently by exploiting the low rank
of the off-diagonal blocks. 

We derive integral formulas for derivatives of solutions of both the
second and fourth order template boundary value problems. In addition,
we consider boundary value problems of the type $Lu=df_{1}/dy$ and
$Lu=d^{2}f_{2}/dy^{2}$ and show how to get the solution $u$ as well
as its derivatives without numerically differentiating $f_{1}$ or
$f_{2}$. In Section 4, these calculations are used to show that numerical
differentiation in the wall-normal or $y$ direction can be entirely
eliminated in the numerical integration of channel flow. 

For the template second order problem, which is $Lu=d^{2}u/dy^{2}-\beta^{2}u=f$
with boundaries $u(\pm1)=0$, we take the Green's function to be $G(y,\eta)=u_{1}(y)v_{1}(\eta)$
for $-1\leq\eta\leq y\leq1$. The Green's function for $-1\leq y\leq\eta\leq1$
is taken to be $G(\eta,y)$ since the problem is symmetric or self-adjoint.
The function $u_{1}$ satisfies the right boundary condition and the
relationship between $u_{1},u_{2}$ and $v_{1},v_{2}$ is as given
in Section 2.1. The Green's function for $-1\leq y\leq\eta\leq1$
is also given by $-u_{2}(y)v_{2}(\eta)$. As a consequence of symmetry,
we have $-u_{2}(\eta)v_{2}(y)=u_{1}(y)v_{1}(\eta)$.

If $G(y,\eta)=u_{1}(y)v_{1}(\eta)$, we have $G_{1}(y,y)=G_{2}(y,y)+1$,
where subscripts of $G$ denote partials with respect to the first
or second argument. This follows from symmetry and Lemma \ref{lem:lemma5}.
In addition, we have $G(1,y)=G(y,-1)=0$ because $u_{1}$ satisfies
the right boundary condition and $v_{1}$ satisfies the left boundary
condition. 

The solution of $Lu=f$ is given by 
\begin{equation}
u(y)=\int_{-1}^{y}G(y,\eta)\, f(\eta)\, d\eta+\int_{y}^{1}G(\eta,y)\, f(\eta)\, d\eta.\label{eq:green2-f2u}
\end{equation}
Differentiating with respect to $y$, we get 
\begin{equation}
u'(y)=\int_{-1}^{y}G_{1}(y,\eta)\, f(\eta)\, d\eta+\int_{y}^{1}G_{2}(\eta,y)\, f(\eta)\, d\eta\label{eq:green2-f2du}
\end{equation}
where subscripts of $G$ stand for partial differentiation. The integral
equation is no longer symmetric in $y$ and $\eta$. Suppose the boundary
value problem is $Lu=df_{1}/dy$. We may substitute $f_{1}'$ for
$f$ in \eqref{eq:green2-f2u} and integrate by parts to get 
\begin{equation}
u(y)=-\int_{-1}^{y}G_{2}(y,\eta)\, f_{1}(\eta)\, d\eta-\int_{y}^{1}G_{1}(\eta,y)\, f_{1}(\eta)\, d\eta.\label{eq:green2-df2u}
\end{equation}
This integral equation for $u(y)$ is not symmetric. Differentiating
with respect to $y$ and using $G_{1}(y,y)=G_{2}(y,y)+1$, we get
\begin{equation}
\frac{du}{dy}=-\int_{-1}^{y}G_{12}(y,\eta)\, f_{1}(\eta)\, d\eta-\int_{y}^{1}G_{12}(\eta,y)\, f_{1}(\eta)\, d\eta+f_{1}(y).\label{eq:green2-df2du}
\end{equation}

The template fourth order problem is $Lu=(D^{2}-\beta^{2})(D^{2}-\alpha^{2})u=f$
with boundary conditions $u(\pm1)=u'(\pm1)=0$. We again take the
Green's function to be $G(y,\eta)$ for $-1\leq\eta\leq y\leq1$.
From Section 2.1, we have $G(y,\eta)=u_{1}(y)v_{1}(\eta)+u_{2}(y)v_{2}(\eta)$.
Figure \ref{fig:EntriesOf4x4Matrix} gives the coefficients of the
Green's function as explained in Section 2.2. The functions $u_{1}(y)$
and $u_{2}(y)$ satisfy the right boundary condition. Using symmetry,
we take the Green's function for $-1\leq y\leq\eta\leq1$ to be $G(\eta,y)$.
As a consequence of symmetry, we have 
\[
u_{1}(y)v_{1}(\eta)+u_{2}(y)v_{2}(\eta)=-u_{3}(\eta)v_{3}(y)-u_{4}(\eta)v_{4}(y).
\]
Using this identity and Lemma \ref{lem:lemma5}, we deduce that 
\begin{align*}
G_{1}(y,y) & =G_{2}(y,y)\\
G_{11}(y,y) & =G_{22}(y,y)\\
G_{111}(y,y) & =G_{222}(y,y)+1.
\end{align*}
Here the subscripts of $G$ denote partial differentiation with $1$
and $2$ standing for the first and second arguments of $G$. Since
$u_{1}$ and $u_{2}$ satisfy the right boundary conditions while
$v_{1}$ and $v_{2}$ satisfy the left boundary conditions, we have
\[
G(1,y)=G_{1}(1,y)=G(y,-1)=G_{2}(y,-1)=0.
\]
In other words, the Green's function satisfies the boundary conditions.

The solution of $Lu=f$, with the boundary conditions associated with
the fourth order template problem, are given by 
\begin{equation}
u(y)=\int_{-1}^{y}G(y,\eta)\, f(\eta)\, d\eta+\int_{y}^{1}G(\eta,y)\, f(\eta)\, d\eta\label{eq:green4-f2u}
\end{equation}
as before. Differentiating with respect to $y$ gives
\begin{align}
\frac{du}{dy} & =\int_{-1}^{y}G_{1}(y,\eta)\, f(\eta)\, d\eta+\int_{y}^{1}G_{2}(\eta,y)\, f(\eta)\, d\eta\label{eq:green4-f2duddu}\\
\frac{d^{2}u}{dy^{2}} & =\int_{-1}^{y}G_{11}(y,\eta)\, f(\eta)\, d\eta+\int_{y}^{1}G_{22}(\eta,y)\, f(\eta)\, d\eta.
\end{align}
The subscripts of $G$ denote differentiation. If the fourth order
template problem is of the form $Lu=df_{1}/dy$, its solution is given
by 
\begin{equation}
u(y)=-\int_{-1}^{y}G_{2}(y,\eta)\, f_{1}(\eta)\, d\eta-\int_{y}^{1}G_{1}(\eta,y)\, f_{1}(\eta)\, d\eta.\label{eq:green4-df2u}
\end{equation}
This form of the solution is obtained after substituting $df_{1}/dt$
for $f$ in \eqref{eq:green4-f2u} and then integrating by parts.
The boundary terms vanish. Differentiating with respect to $y$, we
get 
\begin{align}
\frac{du}{dy} & =-\int_{-1}^{y}G_{12}(y,\eta)\, f_{1}(\eta)\, d\eta-\int_{y}^{1}G_{12}(\eta,y)\, f_{1}(\eta)\, d\eta.\label{eq:green4-df2duddu}\\
\frac{d^{2}u}{dy^{2}} & =-\int_{-1}^{y}G_{112}(y,\eta)\, f_{1}(\eta)\, d\eta-\int_{y}^{1}G_{122}(\eta,y)\, f_{1}(\eta)\, d\eta.
\end{align}
The boundary terms vanish on both occasions. If the template fourth
order boundary value problem is in the form $Lu=d^{2}f_{2}/dy^{2}$,
the analogous formulas are as follows:
\begin{align}
u(y) & =\int_{-1}^{y}G_{22}(y,\eta)\, f_{2}(\eta)\, d\eta+\int_{y}^{1}G_{11}(\eta,y)\, f_{2}(\eta)\, d\eta\nonumber \\
\frac{du}{dy} & =\int_{-1}^{y}G_{122}(y,\eta)\, f_{2}(\eta)\, d\eta+\int_{y}^{1}G_{112}(\eta,y)\, f_{2}(\eta)\, d\eta\nonumber \\
\frac{d^{2}u}{dy^{2}} & =\int_{-1}^{y}G_{1122}(y,\eta)\, f_{2}(\eta)\, d\eta+\int_{y}^{1}G_{1122}(\eta,y)\, f_{2}(\eta)\, d\eta+(G_{122}(y,y)-G_{112}(y,y))\, f_{2}(y).\label{eq:green4-ddf}
\end{align}
These formulas are derived using the properties of $G$ given in the
previous paragraph.

Formulas \eqref{eq:green2-f2u} through \eqref{eq:green4-ddf} give
a method to compute solutions and solution derivatives without numerical
differentiation even when the right hand side of the boundary value
problem is given as a derivative. If the right hand is $df_{1}/dy$,
these formulas use $f_{1}$ and not $df_{1}/dy$. The derivatives
are transferred to the Green's function which can be differentiated
analytically. In the case of the template fourth order problem, the
kernels of the formulas can be described using a matrix such as the
one displayed in Figure \ref{fig:EntriesOf4x4Matrix}. In fact the
kernels can be obtained by multiplying the entries of that matrix
with suitable powers of $\alpha$ and $\beta.$ With such a representation
the kernels can be evaluated in a numerically stable way as described
in Section 2.2.

The numerical evaluation of formulas \eqref{eq:green2-f2u} through
\eqref{eq:green4-ddf} is affected by discretization and rounding
errors in varying ways. To avoid writing down long formulas, we limit
the discussion of numerical errors to the template second order problem
and note that very similar issues arise for the template fourth order
boundary value problem.

When the integral formulation is used, the solution of the template
second order problem $(D^{2}-\beta^{2})u=f$ at the boundary point
$y=1$ is obtained as the sum of the following four terms:
\begin{equation}
\int_{-1}^{1}\frac{\mathrm{e}^{\beta(\eta-1)}f(\eta)}{2\beta(1-{\rm e}^{-4\beta})}\, d\eta,\int_{-1}^{1}\frac{{\rm e}^{-2\beta}{\rm e}^{-\beta(\eta+1)}f(\eta)}{2\beta(1-{\rm e}^{-4\beta})}\, d\eta,-\int_{-1}^{1}\frac{{\rm e}^{-2\beta}{\rm e}^{-\beta(\eta+1)}\, f(\eta)}{2\beta(1-{\rm e}^{-4\beta})}\, d\eta,-\int_{-1}^{1}\frac{{\rm e}^{\beta(\eta-1)}\, f(\eta)}{2\beta(1-{\rm e}^{-4\beta})}\, d\eta.\label{eq:bndrycancelsat1}
\end{equation}
The second and third terms are exceedingly small even for moderate
$\beta$. The main contribution to numerical error is from the exact
cancellation between the first and the last terms. The magnitude of
the first or the last term is of the order $\left|f\right|_{\infty}/\beta^{2}$.
If the quadrature rule is a very good one, each of the integrals may
be evaluated with an error of around $\left|f\right|_{\infty}\beta^{-2}\epsilon_{machine}$.
If such a quadrature rule is devised, the error in the boundary layer
will also be of the same order. If we suppose $f\equiv\beta^{2}$,
then the exact formula will look like $1-{\rm e}^{\beta(y-1)}$ near
the $y=1$. Since $1$ is a special number in machine arithmetic,
the subtraction $y-1$ will be exact at $y=1$ but not at other nearby
points. If we take the subtraction error at $y=1$ to follow the same
model as at other points, we get the error in the boundary layer using
the exact formula to be of the order $\left|1-{\rm e}^{\beta\epsilon}\right|\approx\beta\epsilon$
or $\left|f\right|_{\infty}\epsilon_{machine}/\beta$. Thus the integral
form has the potential to be more accurate in the boundary layer than
even the mathematically exact formula. Here we envisage quadrature
rules for the sort of integrals that occur in \eqref{eq:bndrycancelsat1}
which take into account the occurrence of terms such as ${\rm e}^{-\beta(\eta+1)}$
in the integrands and whose weights and nodes are computed using extended
precision.

The use of formulas such as \eqref{eq:green2-f2du} to compute the
derivative $du/dy$ is especially accurate in the boundary layer.
For instance, at $y=1$ the first term of \eqref{eq:bndrycancelsat1}
gets multiplied by $\beta$ and the last term by $-\beta$ with the
result that there is no cancellation error in the boundary layer.
In view of this observation, some of the errors reported in Table
7 of \cite{GreengardRokhlin1991} may appear a little high for the
function derivative. 

Finally, we consider a type of numerical error that arises in formulas
such as \eqref{eq:green4-df2u} that express the solution of $Lu=df_{1}/dy$
in integral form  without differentiation of $f_{1}$. If $\beta$
is large in the template second order problem $(D^{2}-\beta^{2})u=f$,
the solution satisfies $u\approx-f/\beta^{2}$ away from the boundary.
Thus if $f$ is given in the form $df_{1}/dy$, the solution will
satisfy $u\approx-\frac{df_{1}}{dy}\beta^{-2}$ and a formula such
as \eqref{eq:green4-df2u} essentially has to produce the derivative
of $f_{1}$ away from the boundary using integration. Differentiation
is defined by subtracting nearly equal quantities and the cancellation
errors inherent in that process cannot go away entirely. The same
comment applies to formulas such as \eqref{eq:green2-f2du} which
produce solution derivatives using an integral formula or to the evaluation
of solution derivatives using the background Green's function as in
\cite{GreengardRokhlin1991} or to the method of spectral integration
discussed in the introduction. The principle contribution to the solution
of $(D^{2}-\beta^{2})u=f$ for large $\beta$ and away from the boundary
is due to the term 
\[
\frac{1}{2\beta({\rm e}^{-4\beta}-1)}\int_{-1}^{1}{\rm e}^{-\beta|\eta-y|}\, f(\eta)\, d\eta.
\]
This is the term which makes the solution look like $-f\beta^{-2}$
away from the boundary. When \eqref{eq:green2-f2du} is used to evaluate
$du/dy$ with $u$ being the solution of $Lu=f$, the leading contribution
is from the two terms
\[
\frac{1}{2\beta({\rm e}^{-4\beta}-1)}\left(\int_{-1}^{y}{\rm e}^{\beta(\eta-y)}\, f(\eta)\, d\eta-\int_{y}^{1}{\rm e}^{\beta(y-\eta)}\, f(\eta)\, d\eta\right).
\]
Here the cancellation error we are looking for is evident. A numerical
method that uses integral formulas to evaluate solution derivatives
or to eliminate derivatives that appear on the right hand side would
benefit by treating such terms together, especially when quadrature
rules are derived.

\section{Time integration of the Navier-Stokes equations }

Let ${\bf u}=(u,v,w)$ be the velocity field of channel flow or plane
Couette flow. We assume the domain to be periodic in the wall-parallel
directions with periods equal to $2\pi\Lambda_{x}$ and $2\pi\Lambda_{z}$
in the streamwise and spanwise directions, respectively. The Fourier
decomposition of the velocity field is given by
\[
{\bf u}=\sum_{\ell=-L/2}^{L/2}\sum_{n=N/2}^{N/2}\hat{{\bf u}}_{\ell,n}(y){\rm e}^{i\ell x/\Lambda_{x}+inz/\Lambda_{z}}.
\]
This Fourier decomposition assumes the number grid points in the streamwise
and spanwise directions to be $L$ and $N$. The notation $\hat{{\bf u}}_{\ell,n}$
denotes a Fourier coefficient of the entire velocity field. Similarly,
$\hat{u}_{\ell,n}$, $\hat{v}_{\ell,n}$, $\hat{w}_{\ell,n}$ denote
the Fourier coefficients of the streamwise, wall-normal, and spanwise
components of the velocity field, respectively. The components of
the vorticity $\nabla\times{\bf u}$ are denoted by $\omega_{x},\,\omega_{y},\,\omega_{z}$
and their Fourier components are denoted similarly. 

Often which modes $\ell$ and $n$ apply is clear from context and
the Fourier modes are indicated as $\hat{u}$, $\hat{v}$, $\hat{w}$
without subscripts. The $\ell=n=0$ modes are the mean modes and are
denoted using an over-bar. For example, the mean mode of the streamwise
velocity is $\bar{u}$. In both the flows considered here, the range
of the $y$ variable is $-1\leq y\leq1$, with the walls located at
$y=\pm1$.

\subsection{The Kim-Moin-Moser equations}

We take the Navier-Stokes equations to be $\partial{\bf u}/\partial t+{\bf H}=-\nabla p+\triangle{\bf u}/Re$,
with ${\bf H}=(H_{1},H_{2},H_{3})$ being the nonlinear term. Both
the Kleiser-Schumann \cite{KleiserSchumann1980}and Kim-Moin-Moser
\cite{KimMoinMoser1987} methods begin by substituting the truncated
Fourier expansion of the velocity field ${\bf u}$. The various Fourier
modes are coupled through the nonlinear term. The nonlinear term is
dealiased using the $3/2$ rule \cite{Boyd2001}. 

Both the methods use identical equations for the mean streamwise velocity
and mean spanwise velocity:

\begin{equation}
\frac{\partial\bar{u}}{\partial t}=-\bar{H}_{1}+p_{g}+\frac{1}{Re}\frac{\partial^{2}\bar{u}}{\partial y^{2}}\quad\quad\quad\frac{\partial\bar{w}}{\partial t}=-\bar{H}_{3}+\frac{1}{Re}\frac{\partial^{2}\bar{w}}{\partial y^{2}}\label{eq:kmm-mean-modes}
\end{equation}
For plane Couette flow $p_{g}=0$ and the boundary conditions are
$\bar{u}(\pm1)=\pm1$ and $\bar{w}(\pm1)=0$. For channel flow, $\bar{u}(\pm1)=\bar{w}(\pm1)=0$
but $p_{g}$ is non-zero. We may take $p_{g}=2/Re$ and maintain a
constant pressure gradient or we may take 
\begin{equation}
p_{g}=\frac{1}{2}\int_{-1}^{+1}\bar{H}_{1}\, dy-\frac{1}{2Re}\frac{\partial\bar{u}}{\partial y}\Biggl|_{y=-1}^{y=1}\label{eq:kmm-pg}
\end{equation}
and keep the streamwise mass flux $\frac{1}{2}\int_{-1}^{+1}\bar{u}\, dy$
constant at $2/3$. The laminar solution of channel flow is ${\bf u}=(1-y^{2},0,0)$
in both cases. 

The equations for the $(\ell,n)$ mode are 
\begin{align*}
\frac{\partial\hat{u}}{\partial t}+\hat{H}_{1} & =-\left(\frac{i\ell}{\Lambda_{x}}\right)\hat{p}+\left(D^{2}-\frac{\ell^{2}}{\Lambda_{x}^{2}}-\frac{n^{2}}{\Lambda_{z}^{2}}\right)\hat{u}\\
\frac{\partial\hat{v}}{\partial t}+\hat{H}_{2} & =-\frac{\partial\hat{p}}{\partial y}+\left(D^{2}-\frac{\ell^{2}}{\Lambda_{x}^{2}}-\frac{n^{2}}{\Lambda_{z}^{2}}\right)\hat{v}\\
\frac{\partial\hat{w}}{\partial t}+\hat{H}_{3} & =-\left(\frac{in}{\Lambda_{z}}\right)\hat{p}+\left(D^{2}-\frac{\ell^{2}}{\Lambda_{x}^{2}}-\frac{n^{2}}{\Lambda_{z}^{2}}\right)\hat{w}.
\end{align*}
Here all the hatted variables are Fourier coefficients of the $(\ell,n)$
mode and are functions of $y$. As before $D$ denotes $d/dy$. The
incompressibility constraint $\nabla.{\bf u}=0$ gives $i\ell\hat{u}/\Lambda_{x}+\partial\hat{v}/\partial y+in\hat{w}/\Lambda_{z}=0$.
The equations are solved in this form by the Kleiser-Schumann method.
In the Kim-Moin-Moser method these equations are altered to 
\begin{align}
\frac{\partial\hat{\omega_{y}}}{\partial t}+\hat{H}_{4} & =\frac{1}{Re}\left(D^{2}-\frac{\ell^{2}}{\Lambda_{x}^{2}}-\frac{n^{2}}{\Lambda_{z}^{2}}\right)\hat{\omega_{y}}\nonumber \\
\frac{\partial}{\partial t}\left(D^{2}-\frac{\ell^{2}}{\Lambda_{x}^{2}}-\frac{n^{2}}{\Lambda_{z}^{2}}\right)\hat{v}+\hat{H}_{5} & =\frac{1}{Re}\left(D^{2}-\frac{\ell^{2}}{\Lambda_{x}^{2}}-\frac{n^{2}}{\Lambda_{z}^{2}}\right)^{2}\hat{v}.\label{eq:kmm-etaandv}
\end{align}
The boundary conditions are $\hat{\omega_{y}}(\pm1)=\hat{v}(\pm1)=\frac{d\hat{v}}{dy}(\pm1)=0$.
Here 
\[
H_{4}=\frac{\partial H_{1}}{\partial z}-\frac{\partial H_{3}}{\partial x}\quad\text{and}\quad H_{5}=\frac{\partial^{2}H_{2}}{\partial x^{2}}+\frac{\partial^{2}H_{2}}{\partial z^{2}}-\frac{\partial^{2}H_{1}}{\partial y\partial x}-\frac{\partial^{2}H_{3}}{\partial y\partial z}.
\]
The entire velocity field can be recovered using $\bar{u}$, $\bar{w}$,
$\omega_{y}$, and $v$ \cite{KimMoinMoser1987}. 

Imposing physically correct boundary conditions on pressure causes
some complications and is a potential pitfall. Early discussions of
this issue are found in \cite{KleiserSchumann1980,MoinKim1980}. A
thorough discussion of this topic, important both for mathematical
theory and for computation, is found in an illuminating paper by Rempfer
\cite{Rempfer2006}.

\subsection{Time stepping using Green's functions}

If the Kim-Moin-Moser equations \eqref{eq:kmm-mean-modes} and \eqref{eq:kmm-etaandv}
are discretized in time, we get linear boundary value problems in
the wall-normal or $y$ direction. Green's functions will be used
to solve these linear boundary value problems. An advantage of this
method is that the boundary layers are analytically built into the
Green's functions.

The original paper by Kim, Moin, and Moser \cite{KimMoinMoser1987}
used the CNAB (Crank-Nicolson and Adam-Bashforth) discretization in
time. If the method is applied to the $\hat{\omega}_{y}$ equation
in \eqref{eq:kmm-etaandv}, we get 
\[
\frac{\hat{\omega}_{y}^{n+1}-\hat{\omega}_{y}^{n}}{\Delta t}=-\frac{3\hat{H}_{4}^{n}-\hat{H}_{4}^{n-1}}{2}+\frac{1}{Re}\left(D^{2}-\frac{\ell^{2}}{\Lambda_{x}^{2}}-\frac{n^{2}}{\Lambda_{z}^{2}}\right)\left(\frac{\hat{\omega}_{y}^{n+1}+\hat{\omega}_{y}^{n}}{2}\right).
\]
The superscripts denote time steps. It is well-known that the numerical
stability of Crank-Nicolson can be dicey in spite of its stability
region being the entire left half plane. The eigenvalue equal to $\lambda$
corresponds to an amplification factor of $(1+\lambda\Delta t)/(1-\lambda\Delta t)$.
The amplification is by a factor less than $1$ in magnitude for eigenvalues
with a negative real part. However, the amplification factor can be
very close to $1$ for eigenvalues such as $\lambda=-10^{10}$ which
correspond to rapid decay. If care is taken to use the same scheme
for differentiating $\hat{\omega}_{y}^{n+1}$ and $\hat{\omega}_{y}^{n}$,
or if the boundary value problem is solved for $\hat{\omega}_{y}^{n+1}+\hat{\omega}_{y}^{n}$
at each time step, CNAB will be stable. We found it difficult to stabilize
CNAB for the $\hat{v}$ equation in \eqref{eq:kmm-etaandv}. This
could be because we are mixing integration using a Green's function
with the second derivative that comes from the left hand side of \eqref{eq:kmm-etaandv},
or it could be because the best possible quadrature rules for this
problem are yet to be derived. We will not consider CNAB any further.

Suppose $dX/dt=f(X)+\triangle X/Re$, where $f(X)$ is a nonlinear
term. The time discretizations we consider are of the following form:
\[
\frac{1}{\Delta t}\left(\gamma X^{n+1}+\sum_{j=0}^{s-1}a_{j}X^{n-j}\right)=\sum_{j=0}^{s-1}b_{j}f(X^{n-j})+\frac{1}{Re}\triangle X^{n+1}.
\]
If the $\bar{u}$ equation of \eqref{eq:kmm-mean-modes} is discretized
in time, it fits the template second order boundary value problem
$(D^{2}-\beta^{2})u=f$ with
\begin{equation}
u=\bar{u}^{n+1},\quad\beta^{2}=\frac{\gamma Re}{\Delta t},\quad f=\frac{Re}{\Delta t}\sum_{j=0}^{s-1}a_{j}\bar{u}^{n-j}+Re\sum_{j=0}^{s-1}(\bar{H}_{1}^{n-j}-p_{g}^{n+1}).\label{eq:ubar-bvp}
\end{equation}
The time discretization of the $\bar{w}$ equation of \eqref{eq:kmm-mean-modes}
fits the template second order boundary value problem $(D^{2}-\beta^{2})u=f$
with
\begin{equation}
u=\bar{w}^{n+1},\quad\beta^{2}=\frac{\gamma Re}{\Delta t},\quad f=\frac{Re}{\Delta t}\sum_{j=0}^{s-1}a_{j}\bar{w}^{n-j}+Re\sum_{j=0}^{s-1}H_{3}^{n-j}.\label{eq:wbar-bvp}
\end{equation}
The time discretization of the $\hat{\omega}_{y}$ equation of \prettyref{eq:kmm-etaandv}
also fits the template second order boundary value problem:

\begin{equation}
u=\hat{\omega}_{y}^{n+1},\quad\beta^{2}=\frac{\ell^{2}}{\Lambda_{x}^{2}}+\frac{n^{2}}{\Lambda_{z}^{2}}+\frac{\gamma Re}{\Delta t},\quad f=\frac{Re}{\Delta t}\sum_{j=0}^{s-1}a_{j}\hat{\omega}_{y}^{n-j}+Re\sum_{j=0}^{s-1}b_{j}\hat{H}_{4}^{n-j}.\label{eq:etahat-bvp}
\end{equation}
The time discretization of the $\hat{v}$ equation of \prettyref{eq:kmm-etaandv}
fits the template fourth order boundary value problem \eqref{eq:fourth-order-bvp}
as follows:

\begin{equation}
u=\hat{v}^{n+1},\quad\alpha^{2}=\frac{\ell^{2}}{\Lambda_{x}^{2}}+\frac{n^{2}}{\Lambda_{z}^{2}},\quad\beta^{2}=\alpha^{2}+\frac{\gamma Re}{\Delta t},\quad f=\frac{Re}{\Delta t}\sum_{j=0}^{s-1}a_{j}(D^{2}-\alpha^{2})\hat{v}^{n-j}+Re\sum_{j=0}^{s-1}b_{j}\hat{H}_{5}^{n-j}.\label{eq:vhat-bvp}
\end{equation}
From the manner in which $\alpha^{2}$ and $\beta^{2}$ arise, the
advantage of casting the Green's function for the template fourth
order problem using $\delta=\alpha^{2}-\beta^{2}$ and $\sigma=\alpha^{2}+\beta^{2}$
, as we did in Figure \ref{fig:EntriesOf4x4Matrix} and Section 2.2,
is evident. Both $\delta$ and $\sigma$ can be evaluated without
cancellation errors.

The equations \eqref{eq:ubar-bvp}, \eqref{eq:wbar-bvp}, \eqref{eq:etahat-bvp},
and \eqref{eq:vhat-bvp} are reduced to quadratures of the type \eqref{eq:quadrature-1}
and \eqref{eq:quadrature-2} as explained in Section 2. In the fourth
order problem \eqref{eq:vhat-bvp}, which is solved for $\hat{v}^{n+1}$,
the right hand side uses second derivatives of $\hat{v}^{n-j}$ from
previous stages. The calculation of this second derivative is reduced
to quadratures using \eqref{eq:green4-f2duddu}.

The time stepping schemes we have implemented use
\begin{gather*}
s=1,\,\gamma=1,\, a_{1}=-1,\, b_{1}=1\\
s=2,\,\gamma=3/2,\, a_{1}=-2,\, a_{2}=1/2,\, b_{1}=2,\, b_{2}=-1\\
s=2,\,\gamma=11/6,\, a_{1}=-3,\, a_{2}=3/2,\, a_{3}=-1/3,\, b_{1}=3,\, b_{2}=-3,\, b_{3}=1.
\end{gather*}
These are implicit-explicit multistep schemes that correspond to backward
Euler, BDF2, and BDF3 respectively. For derivations of these schemes,
see \cite{AscherRuuthWetton1995,Crouzeix1980,Varah1980}.

\begin{figure}

\begin{centering}
\includegraphics[scale=0.4]{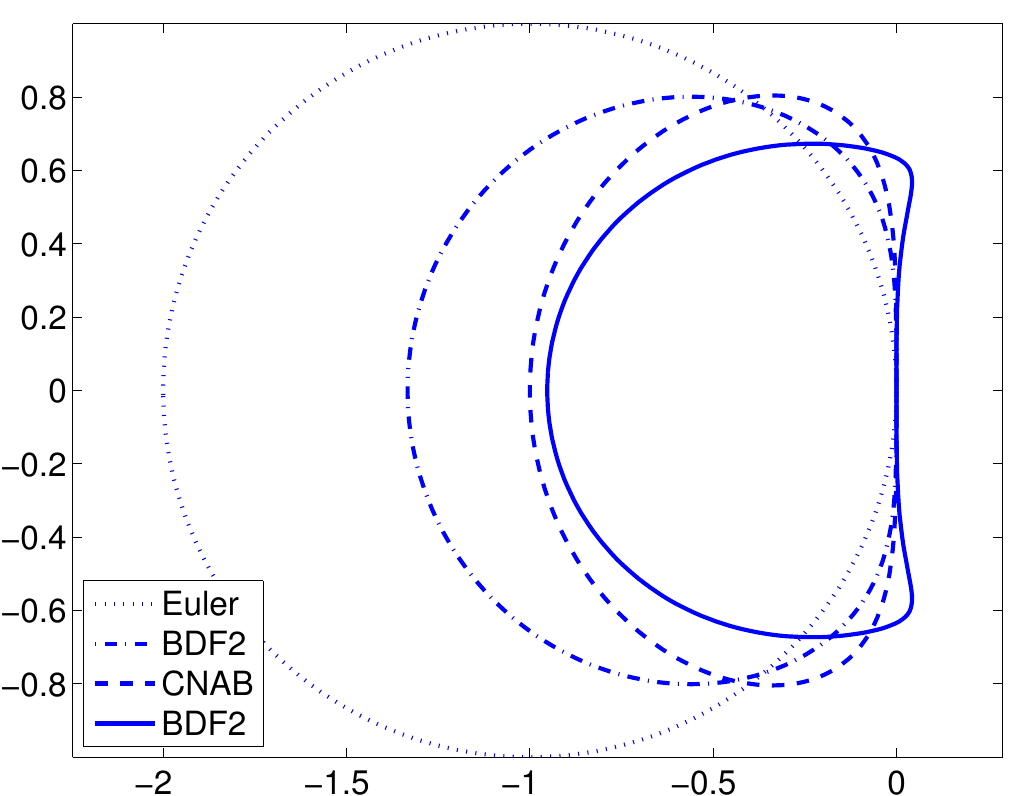}\caption{Absolute stability regions of the explicit halves of the implicit-explicit
methods based on backward Euler, CNAB, BDF2, and BDF3.\label{fig:Absolute-stability-regions}}

\par\end{centering}

\end{figure}

The absolute stability regions of the explicit halves of some time
integration schemes are shown in Figure \ref{fig:Absolute-stability-regions}
for reference. In turbulence simulations, the nonlinear advection
term, which is discretized using an explicit scheme, is more of a
constraint on the time step than the diffusion term which is handled
implicitly. The discretization of the viscous diffusion term is by
itself unconditionally stable. 

To complete the description of these methods, we need to explain the
method that is used to solve the quadrature problems \eqref{eq:quadrature-1}
and \eqref{eq:quadrature-2} numerically. These quadrature problems
are extremely well-conditioned even for large $\mu$. The method that
is currently implemented for \eqref{eq:quadrature-1} expands the
integrand in a Chebyshev series and integrates the terms of the series
using well-known formulas. A better method would be to obtain quadrature
nodes and weights for weighted integrals with weight functions equal
to  ${\rm e}^{\pm\mu(t\mp1)}$, evaluate the other factor $f(t)$
at the nodes using an accurate and efficient interpolation algorithm,
and sum using the quadrature weights. Such a method will be developed
in future research. The current method develops spurious difficulties
when $\mu$ is large, although it is good enough to allow us to exhibit
simulations of fully developed turbulence in Section 5. In addition,
if the idea of representing functions using piecewise Chebyshev collocation,
which is briefly mentioned in the introduction and discussed at greater
length in the context of spectral integration in \cite{Viswanath2012},
is employed, even the basic quadrature that is now implemented is
likely to be adequate, even for very large $\mu$. The Green's functions
of Section 2 are completely independent of the discretization used
in the wall-normal or $y$ direction. The discretization could be
Chebyshev, or piecewise Chebyshev, or something else. The ease with
which piecewise Chebyshev discretization can be incorporated into
numerical methods that use Green's functions was one of our prime
motivations. At the moment, the quadrature problem \eqref{eq:quadrature-2}
is solved numerically using spectral integration. Similar comments
apply to this quadrature problem as well.

In direct numerical simulation of turbulence it is more common to
use explicit-implict Runge-Kutta methods in the direct numerical simulation
of turbulence \cite{HoyasJiminez2006,MoserKimMansour1999,Nikitin2006B,SpalartMoserRogers1991}.
The advantanges of Runge-Kutta are ease of initialization and the
possibility of adaptive time-stepping with embedded pairs. In the
solution of ordinary differential equations, multistep methods and
Runge-Kutta methods have been compared extensively \cite{HairerNorsettWanner}.
It is now known that multistep methods can be initialized and time
stepped adatively with equal effectiveness. This technology will be
carried over to implicit-explicit multistep formulas in future research.
Here we have preferred multistep formulas partly in order to leave
room for this future research. Runge-Kutta methods typically have
larger stability region but with each step costing more function evaluations.
In future research, we will show how to derive implicit-explicit multistep
methods with stability regions that are particularly advantageous
for problems such as high $Re$ turbulence simulations.

\section{A discrete model without spatial differentiation}

Here we explain how numerical differentiation in the wall-normal or
$y$ direction can be completely eliminated by employing the divergence
form of the nonlinear term. 

The Kim-Moin-Moser equation for $\bar{u}$ given in \eqref{eq:kmm-mean-modes}
has an $\bar{H}_{1}$ term. The nonlinear term $H_{1}$ is given by
$H_{1}=\partial_{x}(u^{2})+\partial_{y}(uv)+\partial_{z}(uw)$. So
we may take $\bar{H}_{1}=\partial_{y}\bar{uv}$. The mean mode $\bar{u}$
is advanced in time by solving the boundary value problem \eqref{eq:ubar-bvp}.
The right hand side of the template boundary value problem is taken
to be $f$, where $f$ is given by \eqref{eq:ubar-bvp}. The $\bar{H_{1}}$
terms in that right hand side may be removed and a new right hand
side written as $df_{1}/dy$ introduced in their place. The contribution
of a given time step to $f_{1}$ is taken to be $Re$ times $\bar{uv}$
evaluated at that time step. These contributions are weighted by $b_{j}$
and combined as before. The mean mode $\bar{u}$ is advanced in time
using \eqref{eq:green2-f2u} and \eqref{eq:green4-df2u} and its derivative,
if needed in \eqref{eq:kmm-pg}, is calculated using \eqref{eq:green2-f2du}
and \eqref{eq:green2-df2du}. 

The $\bar{w}$ and $\hat{\omega}_{y}$ equations are treated similarly.
The $\hat{\omega}_{y}$ equation in \eqref{eq:kmm-etaandv} has an
$\hat{H}_{4}$ term. The terms of $H_{4}$ that do not require differentiation
in $y$ are 
\[
\frac{\partial^{2}u^{2}}{\partial x\partial z}+\frac{\partial^{2}uw}{\partial z^{2}}-\frac{\partial^{2}uw}{\partial x^{2}}-\frac{\partial^{2}w^{2}}{\partial x\partial z}
\]
and the terms which require a single differentiation are 
\[
\frac{\partial^{2}uv}{\partial y\partial z}-\frac{\partial^{2}vw}{\partial x\partial y}.
\]
These terms are separated and some of them are removed from the $f$
given in \eqref{eq:etahat-bvp} and a new term $df_{1}/dy$ is inserted
in the right hand side. 

The treatment of $\hat{v}$ modes is a bit more elaborate. The right
hand side $H_{5}$ may be written as
\begin{multline*}
\left(\frac{\partial^{3}(uv)}{\partial x^{3}}+\frac{\partial^{3}(uv)}{\partial x\partial z^{2}}+\frac{\partial^{3}(vw)}{\partial x^{2}\partial z}+\frac{\partial^{3}(vw)}{\partial z^{3}}\right)+\left(\frac{\partial^{3}(v^{2})}{\partial x^{2}\partial y}+\frac{\partial^{3}(v^{2})}{\partial y\partial z^{2}}-\frac{\partial^{3}(u^{2})}{\partial x^{2}\partial y}-2\frac{\partial^{3}(uw)}{\partial x\partial y\partial z}-\frac{\partial^{3}(w^{2})}{\partial y\partial z^{2}}\right)\\
+\left(-\frac{\partial^{3}(uv)}{\partial x\partial y^{2}}-\frac{\partial^{3}(vw)}{\partial y^{2}\partial z}\right),
\end{multline*}
where terms are grouped depending upon whether they require zero,
one, or two differentiations with respect to $y.$ The right hand
side of the template fourth order problem which is given as $f$ in
\eqref{eq:vhat-bvp} may be rewritten as $f+df_{1}/dy+d^{2}f_{2}/dy^{2}$,
with none of $f$, $f_{1}$, and $f_{2}$ involving differentiation
with respect to $y$. The integral equations \eqref{eq:green4-f2u}
through \eqref{eq:green4-ddf} may be used to produce $\hat{v}$,
$d\hat{v}/dy$, and even $d^{2}\hat{v}/dy^{2}$. The first derivative
$d\hat{v}/dy$ is needed when the full velocity field is reconstructed
from $\bar{u}$, $\bar{w}$, and the modes $\hat{\omega}_{y}$ and
$\hat{v}$ \cite{KimMoinMoser1987}. Turning this into a practical
method hinges on numerical issues discussed at the end of Section
2. Eliminating numerical differentiation with respect to $y$ may
be useful if a large number of Chebyshev points is used in the $y$
direction. However, it appears that piecewise Chebyshev grids can
resolve boundary layers and internal layers while using only a small
number of Chebyshev points in each sub-interval. Handling piecewise
Chebyshev grids after reducing each time step to quadratures of the
form \eqref{eq:quadrature-1} and \eqref{eq:quadrature-2} is as easy
as $\int_{a}^{c}=\int_{a}^{b}+\int_{b}^{c}$. In piecewise Chebyshev
grids, numerical errors due to differentiation are not a cause for
concern.

\section{Numerical validation}

The numerical method described in Section 3 for solving the incompressible
Navier-Stokes equations has been implemented and tested in numerous
ways. Many earlier computations of plane Couette flow and channel
flow have been reproduced with precision. In this section, we describe
a few computations of fully developed turbulence. All the computations
described here are for channel flow. Channel flow is used far more
often than plane Couette flow in turbulence simulations.

A useful summary of turbulence computations of channel flow is given
by Toh and Itano \cite{TohItano2005}. The Reynolds number $Re$ by
itself is not a good metric to assess the difficulty of a turbulence
computation because simple solutions such as the laminar solution
can be computed easily at any Reynolds number. The metric must take
into account both the Reynolds number and the kind of solutions that
the simulation generates. One useful metric is obtained by taking
the time average of $d\bar{u}/dy$ at the walls, where $\bar{u}$
is the mean streamwise velocity, and then computing $Re_{\tau}=\sqrt{Re\times|d\bar{u}/dy|}$.
The frictional Reynolds number $Re_{\tau}$ is a good measure of the
difficulty of the simulation. The highest $Re_{\tau}$ reached appears
to be $2000$ in the work of Hoyas and Jimenez\cite{HoyasJiminez2006}.
The lowest $Re_{\tau}$ at which one still observe turbulence appears
to be around $100$ \cite{JimenezMoin1991}. Nikitin \cite{Nikitin2006,Nikitin2006B}
has derived a method for solving the incompressible Navier-Stokes
equations in orthogonal curvilinear coordinates. Nikitin's method
uses staggered grids, centered differences, cell averages for nonlinear
terms, and explicit projections to enforce the incompressibility condition.
The same program can handle channel, pipe, eccentric pipe and other
geometries. Nikitin's method has been used to simulate fully developed
turbulence at an $Re_{\tau}$ of $500$. 

\begin{figure}
\begin{centering}
\includegraphics[scale=0.35]{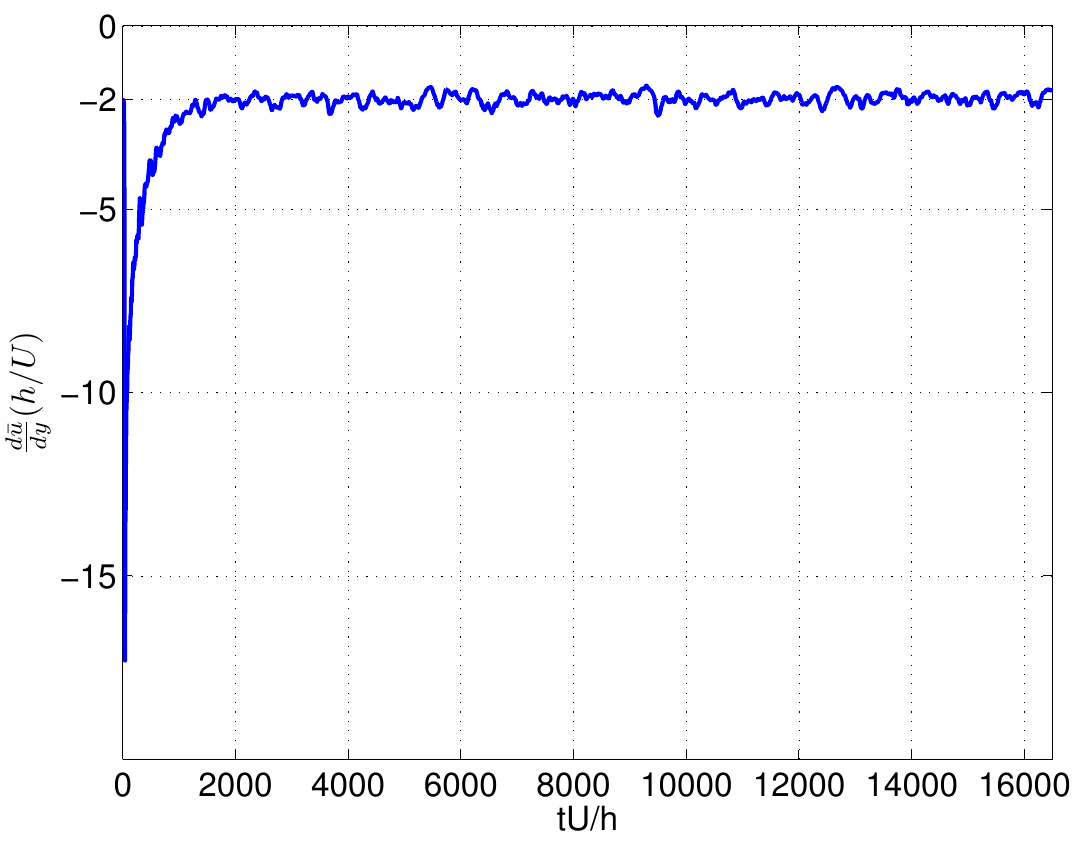}\qquad{}\includegraphics[scale=0.35]{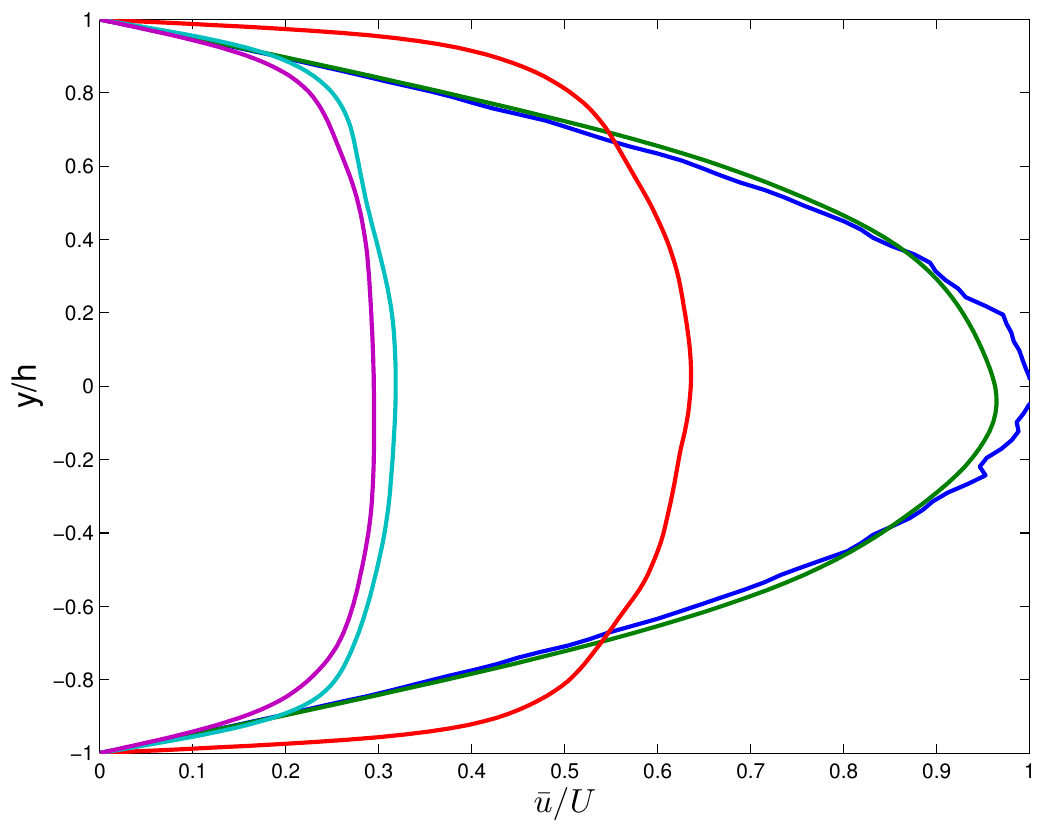}\caption{The first plot shows the variation of the mean shear at the wall as
a function of time. The second plot shows the evolution of the mean
flow during transition to turbulence. Units are given in terms of
channel half-width $h$ and laminar center-line velocity $U$. \label{fig:pressure-bndry-cond}}

\par\end{centering}

\end{figure}
Prior to numerical validation, we discuss the pressure boundary condition
to make a point about the behavior of $\bar{H}_{1}$ in a turbulent
flow. Figure \ref{fig:pressure-bndry-cond} shows a simulation of
channel flow with $Re=10^{4}$, $\Lambda_{x}=2.0$ and $\Lambda_{z}=1.0$.
The grid parameters used $L=64$, $M=128$, and $N=64$. The boundary
condition used was $p_{g}=2/Re$ in the equation for $\bar{u}$ given
in \eqref{eq:kmm-mean-modes}. It is noticeable that the mean shear
converges to $-2$ at the upper wall. The equation for the mean flow
$\bar{u}$ is given by $\partial\bar{u}/\partial t+\bar{H}_{1}=p_{g}+\frac{1}{Re}\partial^{2}\bar{u}/\partial y^{2}$.
The mean flow fluctuates very little once the flow is fully turbulent.
If we average the mean flow over time, we get $\partial^{2}\bar{u}/\partial y^{2}=Re(\bar{H}_{1}-p_{g})$
with $\bar{u}(\pm1)=0$. The Green's function for this boundary value
problem is given by $G(y,t)=(y-1)(\eta+1)/2$ for $-1\leq\eta\leq y\leq1$.
Using \eqref{eq:green2-f2du}, we get 
\[
\frac{\partial\bar{u}}{\partial y}\Biggl|_{y=+1}=\int_{-1}^{1}G_{1}(y,\eta)f(\eta)\, d\eta=Re\int_{-1}^{1}\frac{(\eta-1)(\bar{H}_{1}-p_{g})}{2}\, d\eta=-2.
\]
Since $p_{g}=2/Re$, we have $\int_{-1}^{1}\bar{H}_{1}(y)\, dy=0$.
The reason $\bar{H}_{1}$ satisfies this condition appears not to
be known. In rotating channel flows, the mean flow exhibits a stretch
where its slope is given by the rate of rotation (see Figure 3 of
Yang and Wu \cite{YangWu2012}). The reason for that phenomenon too
appears to be unknown. The mean flow flattens during transition as
evident from the second plot of Figure \ref{fig:pressure-bndry-cond}.
At the very beginning, the mean flow develops oscillations which look
somewhat like the oscillatory shears considered in \cite{LiLin2011}.
It is well-know that fixing the mass flux leads to a quite different
value for the mean shear. Figure \ref{fig:Similar-to-Figure} shows
a turbulence simulation which fixed the mass-flux using \eqref{eq:kmm-pg}.
The parameters used were $Re=10^{4}$, $\Lambda_{x}=1.0$ and $\Lambda_{z}=0.5$.
The grid parameters used $L=256$, $M=256$, and $N=128$ correspond
to approximately $8.5$ million grid points. 

\begin{figure}
\begin{centering}
\includegraphics[scale=0.35]{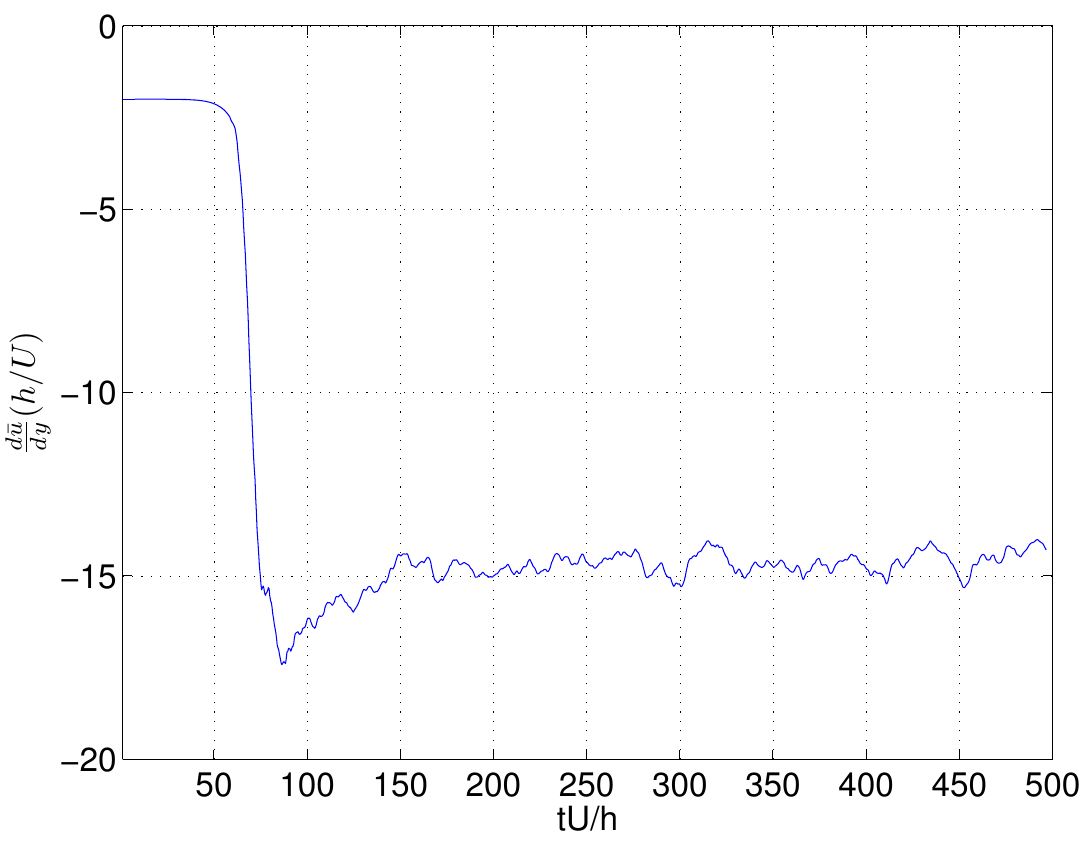}\qquad{}\includegraphics[scale=0.35]{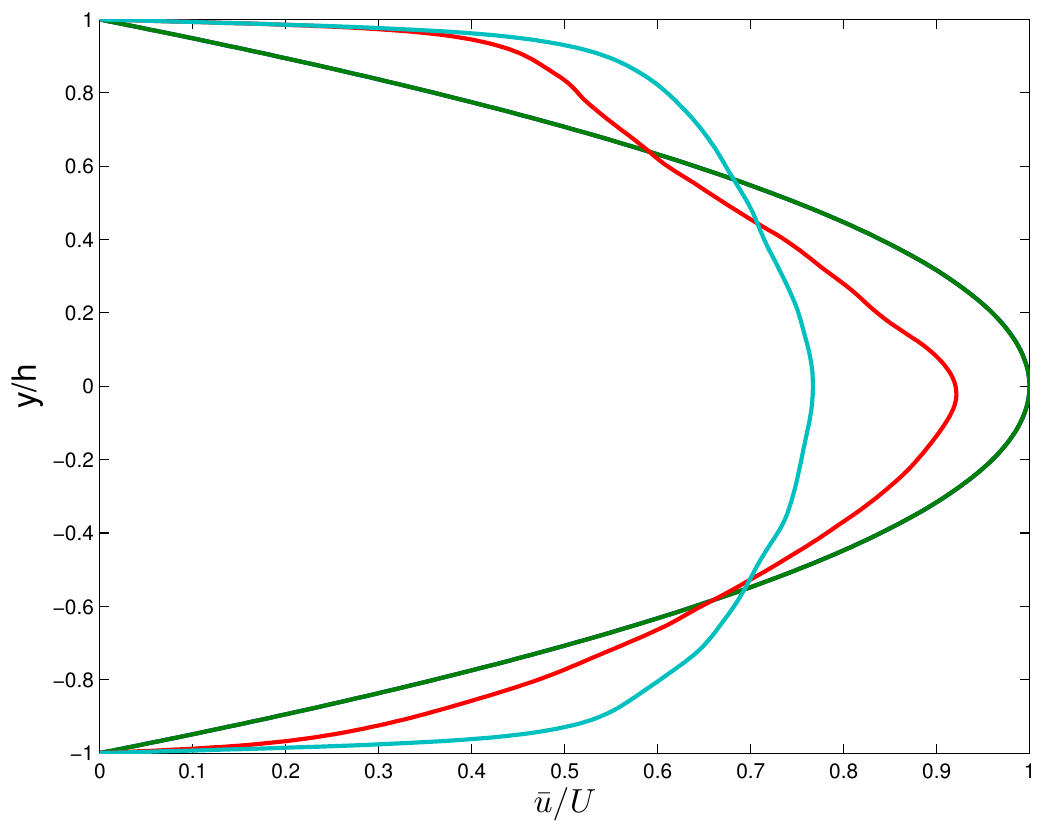}
\par\end{centering}

\caption{Similar to Figure \ref{fig:pressure-bndry-cond} but using a boundary
condition that fixes the mass flux. \label{fig:Similar-to-Figure}}
\end{figure}

\begin{table}
\begin{centering}
\begin{tabular}{|c|c|c|c|c|c|c|c|c|c|c|}
\hline 
$Re$ & $Re_{\tau}$ & $\Lambda_{x}$ & $\Lambda_{z}$ & $L$ & $M$ & $N$ & $\Delta x^{+}$/$\Delta y_{c}^{+}$/$\Delta z^{+}$ & $UT_{1}/h$ & $UT_{2}/h$ & $U\Delta t/h$\tabularnewline
\hline 
\hline 
$4000$ & $171.4$ & $2$ & $2/3$ & $320$ & $128$ & $128$ & $6.7$/$4.2$/$5.6$ & $1125$ & $300$ & $.005$\tabularnewline
\hline 
$10^{4}$ & $379.9$ & $1$ & $1/2$ & $320$ & $256$ & $256$ & $7.4$/$4.6$/$4.6$ & $580$ & $120$ & $.002$\tabularnewline
\hline 
\end{tabular}
\par\end{centering}

\caption{Parameters of turbulence simulations used for numerical validation.
The $\Delta x^{+}/\Delta y_{c}^{+}/\Delta z^{+}$ column gives the
grid resolution in wall units, with $\Delta y_{c}^{+}$ being the
distance between grid points along $y$ near the center of the channel.
The initialization time used to eliminate transients $T_{1}$, the
time over statistics is accumulated $T_{2}$, and the time step $\Delta t$
are non-dimensionalized using the centerline speed of the laminar
solution $U$ and the channel half-width $h$.\label{tab:sec5-validation}}
\end{table}

\begin{figure}
\subfloat[]{

\raggedleft{}\includegraphics[scale=0.4]{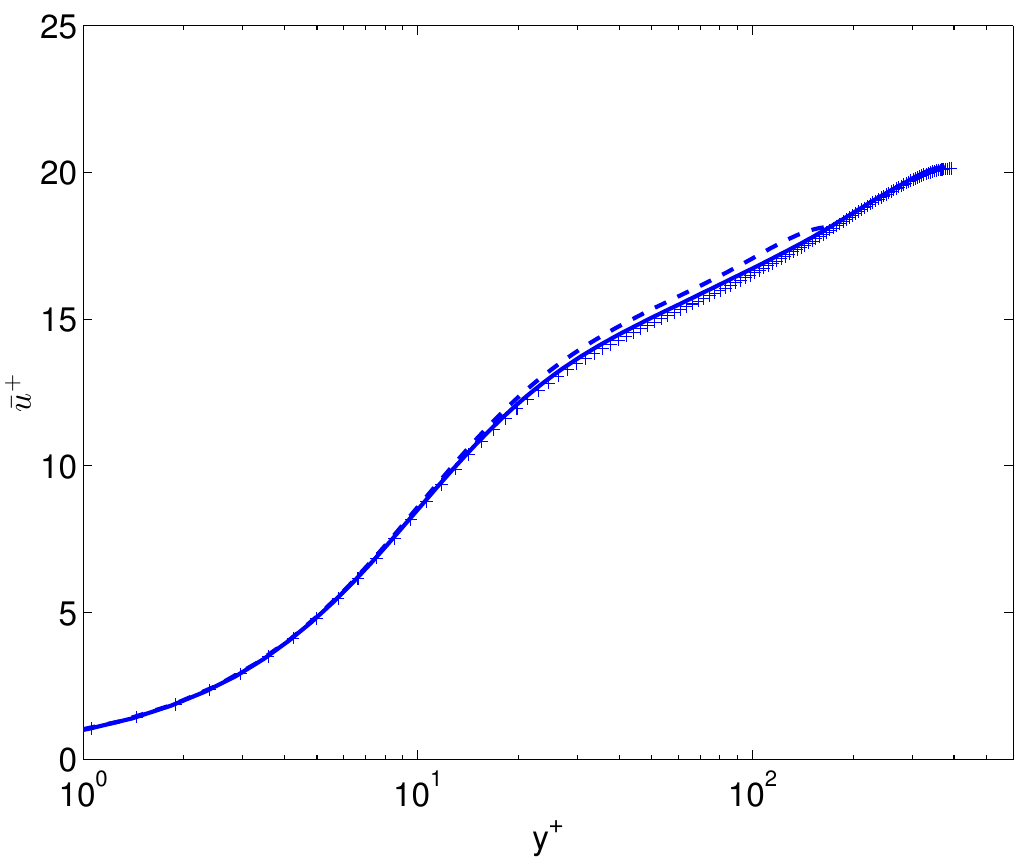}}\subfloat[]{\begin{centering}
\includegraphics[scale=0.4]{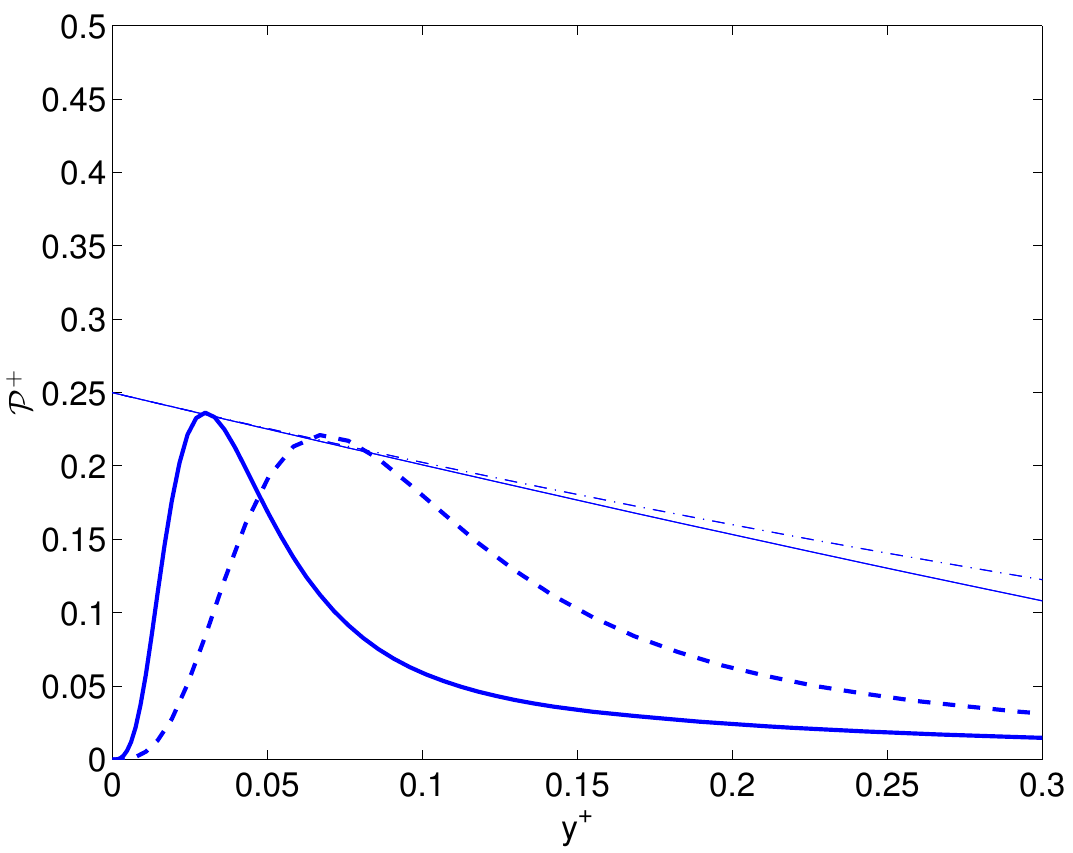}
\par\end{centering}

}

\subfloat[]{

\centering{}\includegraphics[scale=0.4]{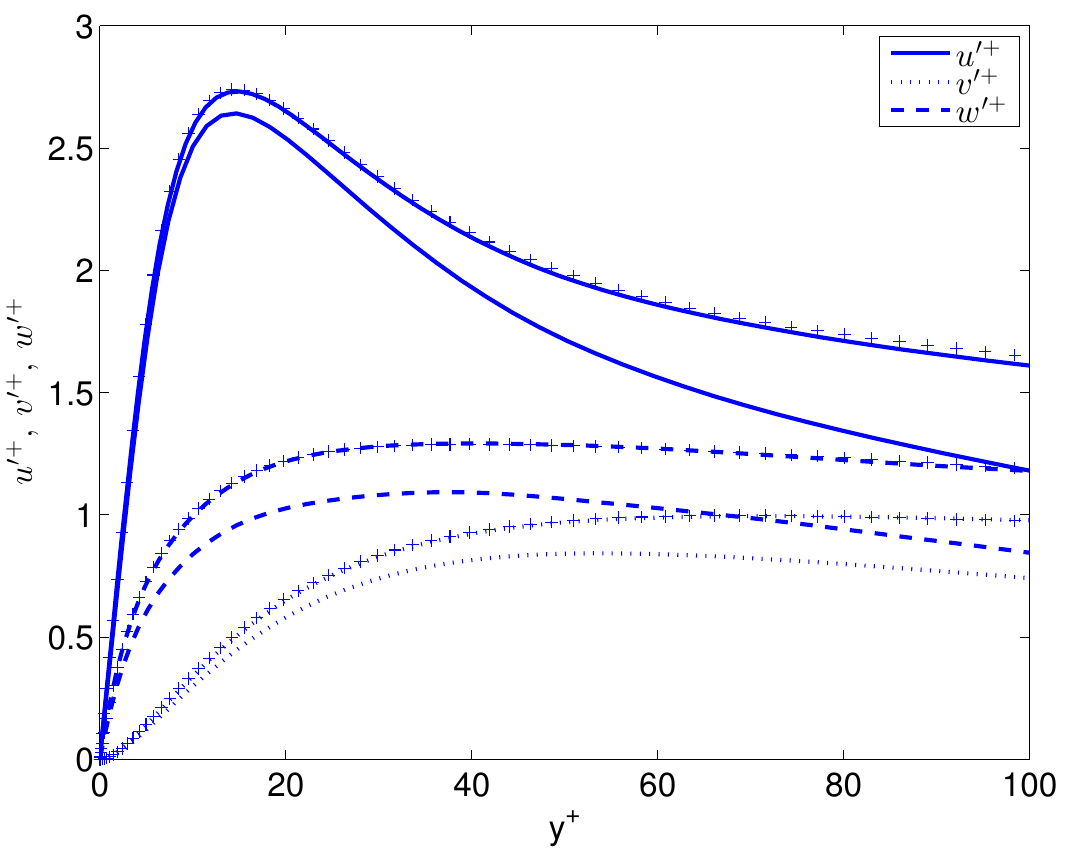}}\subfloat[]{

\begin{centering}
\includegraphics[scale=0.4]{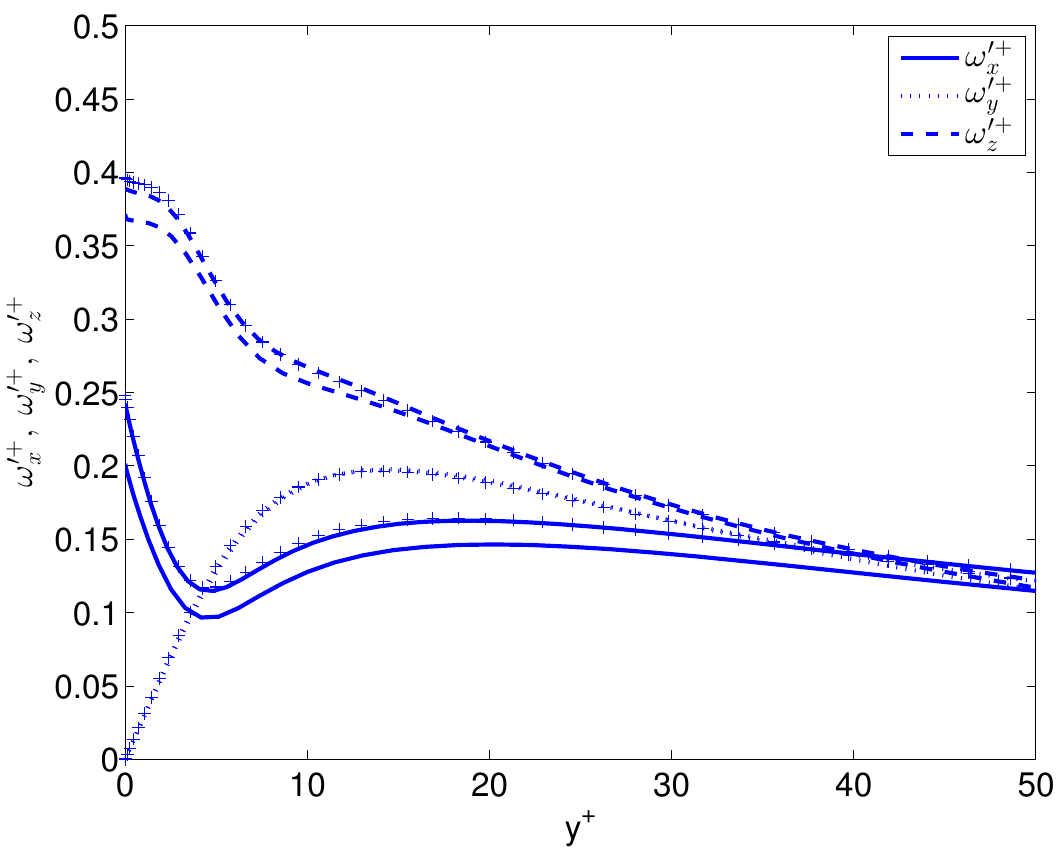}
\par\end{centering}

}

\caption{Statistics for the two cases detailed in Table \ref{tab:sec5-validation}.
In all the figures, the horizontal axis is $y^{+}$, which denotes
the distance from the wall in frictional units. In plots (a) and (b),
the dashed lines are for $Re_{\tau}=171$, which is the first row
in the table, and the solid line is for the other row with $Re_{\tau}=380$.
The data in these plots is compared to classical simulations as well
as existing theory in the text. Data from the $Re_{\tau}=392$ simulation
by Moser, Kim and Mansour \cite{MoserKimMansour1999} is marked using
$+$.\label{fig:sec5-validation}}
\end{figure}
The two test cases used for numerical validation are detailed in Table
\ref{tab:sec5-validation}. The two cases are close to but not exactly
the same as two of the cases reported in the simulations of Moser
et al. \cite{MoserKimMansour1999}. The first test case has $Re_{\tau}=171$
against $Re_{\tau}=180$ in \cite{MoserKimMansour1999}. The second
test case has $Re_{\tau}=380$ against $Re_{\tau}=392$ in \cite{MoserKimMansour1999}. 

The plots in Figure \ref{fig:sec5-validation} may be compared with
the plots of Moser et al. \cite{MoserKimMansour1999}. Figure \ref{fig:sec5-validation}
(a) shows the mean streamwise velocity as a function of the distance
from the wall. This curve is usually fitted using the famous log-law
of the wall. Moser et al. discuss power law fits as well. The plot
for $Re_{\tau}=380$ is quite close to the plot for $Re_{\tau}=392$
given in \cite{MoserKimMansour1999}. The plot for $Re_{\tau}=171$
shows a pronounced low Reynolds number effect. However, the low Reynolds
number effect is much less pronounced in our test than in the $Re_{\tau}=180$
test in \cite{MoserKimMansour1999}. This could be either because
of the superior resolution in our test or because of the much longer
time integration used to eliminated transients.

Figure \ref{fig:sec5-validation} (b) shows the turbulence energy
production as a function of the distance from the wall. The straight
lines in the figure are theoretical curves for the peak (solid line)
and the envelope (dashed line) derived by Laadhari \cite{Laadhari2002}.
Both test cases are in excellent agreement with Laadhari's theory
and simulations.

Figure \ref{fig:sec5-validation} (c) and (d) show turbulence intensities
and rms vorticity profiles, respectively. The plots for $Re_{\tau}=380$
are in good agreement with plots for $Re_{\tau}=392$ reported in
\cite{MoserKimMansour1999}. There is a noticeable discrepancy in
the rms plots of $\omega_{x}$ between our $Re_{\tau}=171$ test case
and the $Re_{\tau}=180$ test case of \cite{MoserKimMansour1999}.
This discrepancy is most probably due to the superior resolution of
our test case. Close observation of the plots in Figure \ref{fig:sec5-validation}
shows that the mean streamwise velocity $\bar{u}$ of the $Re_{\tau}=380$
test case is slightly above that of $Re_{\tau}=392$ benchmark (marked
using $+$) while the turbulence intensities as indicated by the rms
velocity $u'$ are more noticeably lower. The slight differences,
especially the later, are probably because of the lower value of the
frictional Reynolds number in the test case as against the benchmark.
The time interval used to gather statistics is another factor which
causes slight variations.

\section{Conclusion}

Green's function based methods are known to be advantageous for resolving
boundary layers. Since boundary layers become thinner as the Reynolds
number increases, it is reasonable to try Green's function based methods
for fully developed turbulence. In this paper, we have worked out
a numerical method based on Green's functions for channel flow and
plane Couette flow and demonstrated that it is capable of reproducing
turbulence phenomena correctly.

Current methods for turbulent channel flow and plane Couette flow
are the result of intensive research spanning more than three decades.
The Green's function approach developed here builds upon that research
at many points. Green's functions have not been shown to work for
nonlinear problems of the complexity of fully developed turbulence
requiring tens of millions of grid points. Getting the Green's function
approach to work for fully developed turbulence is a task in itself. 

In the sequel to this paper, the method will be extended in two different
directions. The first direction will use a piecewise Chebyshev grid
in the wall-normal direction as well as spectral integration, which
amounts to implicit use of Green's functions. The second direction
will use a piecewise Chebyshev grid, explicit Green's functions as
developed in this paper, and carefully derived quadrature rules. Both
 directions appear capable of challenging or going beyond the current
state of the art in turbulence simulations.

\section{Acknowledgments }

\lyxaddress{We thank Sergei Chernyshenko and Nikolay Nikitin for helpful discussions.
We thank Fabian Waleffe and the referees for many valuable suggestions.
This research was partially supported by NSF grants DMS-0715510, DMS-1115277,
and SCREMS-1026317.}

\bibliographystyle{plain}
\bibliography{references}

\begin{thebibliography}{10}

\bibitem{AscherRuuthWetton1995}
U.M. Ascher, S.J. Ruuth, and B.T.R. Wetton.
\newblock Implicit-explicit methods for time-dependent partial differential
  equations.
\newblock {\em SIAM Journal on Numerical Analysis}, 32:797--823, 1995.

\bibitem{Boyd2001}
J.P. Boyd.
\newblock {\em Chebyshev and Fourier Spectral Methods}.
\newblock Dover, 2001.

\bibitem{CharalambidesWaleffe2008}
M.~Charalambides and F.~Waleffe.
\newblock Gegenbauer tau methods with and without spurious eigenvalues.
\newblock {\em SIAM J. on Numerical Analysis}, 47(1):48--68, 2008.

\bibitem{CoddingtonLevinson1955}
E.A. Coddington and N.~Levinson.
\newblock {\em Theory of Ordinary Differential Equations}.
\newblock McGraw-Hill, 1955.

\bibitem{CoutsiasHT1996}
E.A. Coutsias, T.~Hagstrom, and D.~Torres.
\newblock An efficient spectral method for ordinary differential equations with
  rational function coefficients.
\newblock {\em Mathematics of Computation}, 65(214):611--636, 1996.

\bibitem{Crouzeix1980}
M.~Crouzeix.
\newblock Une m{\'e}thode multipas implicite-explicite pour l'approximation des
  {\'e}quations d'{\'e}volution paraboliques.
\newblock {\em Numerische Mathematik}, 35(3):257--276, 1980.

\bibitem{GottliebOrszag1977}
D.~Gottlieb and S.A. Orszag.
\newblock {\em Numerical Analysis of Spectral Methods: Theory and
  Applications}.
\newblock Society for Industrial and Applied Mathematics, 1993.

\bibitem{Greengard1991}
L.~Greengard.
\newblock Spectral integration and two-point boundary value problems.
\newblock {\em SIAM Journal on Numerical Analysis}, 28:1071--1080, 1991.

\bibitem{GreengardRokhlin1991}
L.~Greengard and V.~Rokhlin.
\newblock On the numerical solution of two-point boundary value problems.
\newblock {\em Communications on Pure and Applied Mathematics}, 44(4):419--452,
  1991.

\bibitem{HairerNorsettWanner}
E.~Hairer, S.P. Norsett, and G.~Wanner.
\newblock {\em Solving Ordinary Differential Equations I}.
\newblock Springer-Verlag, 2nd edition, 2011.

\bibitem{HoyasJiminez2006}
S.~Hoyas and J.~Jim{\'e}nez.
\newblock Scaling of the velocity fluctuations in turbulent channels up to
  {$Re_\tau = 2003$}.
\newblock {\em Physics of Fluids}, 18:011702, 2006.

\bibitem{JimenezMoin1991}
J.~Jimenez and P.~Moin.
\newblock The minimal flow unit in near-wall turbulence.
\newblock {\em Journal of Fluid Mechanics}, 225(213-240), 1991.

\bibitem{KimMoinMoser1987}
J.~Kim, P.~Moin, and R.~Moser.
\newblock Turbulence statistics in fully developed channel flow at low
  {Reynolds} number.
\newblock {\em Journal of Fluid Mechanics}, 177(1):133--166, 1987.

\bibitem{KleiserSchumann1980}
L.~Kleiser and U.~Schumann.
\newblock Treatment of incompressibility and boundary conditions in {3-D}
  numerical spectral simulations of plane channel flows.
\newblock In {\em 3rd Conference on Numerical Methods in Fluid Mechanics},
  volume~1, pages 165--173, 1980.

\bibitem{Laadhari2002}
F.~Laadhari.
\newblock On the evolution of the maximum turbulent kinetic energy production
  in a channel flow.
\newblock {\em Physics of Fluids}, 14:L65--L68, 2002.

\bibitem{Ladyzhenskaya1969}
O.~A. Ladyzhenskaya.
\newblock {\em The Mathematical Theory of Viscous Incompressible Flow, 2nd ed.}
\newblock Gordon and Breach, New York, 1969.

\bibitem{LiLin2011}
Y.C. Li and Z.~Lin.
\newblock A resolution of the {Sommerfeld} paradox.
\newblock {\em SIAM Journal on Mathematical Analysis}, 43:1923, 2011.

\bibitem{LundbladhHenningsonJohansson1992}
A.~Lundbladh, D.S. Henningson, and A.V. Johansson.
\newblock An efficient spectral integration method for the solution of the
  {Navier-Stokes} equations.
\newblock {\em Technical Report FFA TN}, 1992-28, 1992.

\bibitem{LundbladhHenningsonReddy1994}
A.~Lundbladh, D.S. Henningson, and S.C. Reddy.
\newblock Threshold amplitudes for transition in channel flows.
\newblock {\em Transition, Turbulence and Combustion}, pages 309--318, 1994.

\bibitem{MoinKim1980}
P.~Moin and J.~Kim.
\newblock On the numerical solution of time-dependent viscous incompressible
  fluid flows involving solid boundaries.
\newblock {\em Journal of Computational Physics}, 35(3):381--392, 1980.

\bibitem{MoserKimMansour1999}
R.D. Moser, J.~Kim, and N.N. Mansour.
\newblock Direct numerical simulation of turbulent channel flow up to {$Re_\tau
  = 590$}.
\newblock {\em Physics of Fluids}, 11:943--945, 1999.

\bibitem{Nikitin2006}
N.~Nikitin.
\newblock Finite-difference method for incompressible {Navier--Stokes}
  equations in arbitrary orthogonal curvilinear coordinates.
\newblock {\em Journal of Computational Physics}, 217(2):759--781, 2006.

\bibitem{Nikitin2006B}
N.~Nikitin.
\newblock Third-order-accurate semi-implicit {Runge--Kutta} scheme for
  incompressible {Navier--Stokes} equations.
\newblock {\em International Journal for Numerical Methods in Fluids},
  51(2):221--233, 2006.

\bibitem{Rempfer2006}
D.~Rempfer.
\newblock On boundary conditions for incompressible {Navier-Stokes} problems.
\newblock {\em Applied Mechanics Reviews}, 59:107, 2006.

\bibitem{Rokhlin1983}
V.~Rokhlin.
\newblock Solution of acoustic scattering problems by means of second kind
  integral equations.
\newblock {\em Wave Motion}, 5(3):257--272, 1983.

\bibitem{Rokhlin1985}
V.~Rokhlin.
\newblock Rapid solution of integral equations of classical potential theory.
\newblock {\em Journal of Computational Physics}, 60(2):187--207, 1985.

\bibitem{SpalartMoserRogers1991}
P.R. Spalart, R.D. Moser, and M.M. Rogers.
\newblock Spectral methods for the navier-stokes equations with one infinite
  and two periodic directions.
\newblock {\em Journal of Computational Physics}, 96:297--324, 1991.

\bibitem{StarrRokhlin1994}
P.~Starr and V.~Rokhlin.
\newblock On the numerical solution of two-point boundary value problems {II}.
\newblock {\em Communications on Pure and Applied Mathematics},
  47(8):1117--1159, 1994.

\bibitem{TohItano2005}
S.~Toh and T.~Itano.
\newblock Interaction between a large-scale structure and near-wall structures
  in channel flow.
\newblock {\em Journal of Fluid Mechanics}, 524:249--262, 2005.

\bibitem{Varah1980}
J.M. Varah.
\newblock Stability restrictions on second order, three level finite difference
  schemes for parabolic equations.
\newblock {\em SIAM Journal on Numerical Analysis}, 17:300--309, 1980.

\bibitem{Viswanath2012}
D.~Viswanath.
\newblock Spectral integration of linear boundary value problems.
\newblock {\em arxiv.org}, 2012.

\bibitem{Waleffe2003}
F.~Waleffe.
\newblock Homotopy of exact coherent structures in plane shear flows.
\newblock {\em Physics of Fluids}, 15:1517, 2003.

\bibitem{YangWu2012}
Y.T. Yang and J.Z. Wu.
\newblock Channel turbulence with spanwise rotation studied using helical wave
  decomposition.
\newblock {\em Journal of Fluid Mechanics}, 692:137, 2012.

\bibitem{Zebib1984}
A.~Zebib.
\newblock A {Chebyshev} method for the solution of boundary value problems.
\newblock {\em Journal of Computational Physics}, 53(3):443--455, 1984.

\end{thebibliography}

\end{document}